\documentclass[12pt]{article}
\usepackage[margin=1.0in]{geometry}

\usepackage{graphicx}
\usepackage{amsthm,amsmath,amssymb, amsfonts}
\usepackage{color}

\usepackage[T1]{fontenc}
\usepackage{charter}

\usepackage{setspace}
\usepackage{color, xcolor}
\definecolor{myblue}{RGB}{94, 129, 181}
\definecolor{myorange}{RGB}{225, 156, 36}
\definecolor{mygreen}{RGB}{143, 176, 50}

\def\E{\mathbb{E}}
\def\bF{\mathbb{F}}

\def\bP{\mathbb P}

\def\R{\mathbb{R}}

\def\sA{{\mathcal A}}

\def\sE{{\mathcal E}}
\def\sF{{\mathcal F}}

\def\sL{{\mathcal L}}
\def\sM{{\mathcal M}}

\def\sS{{\mathcal S}}

\def\eps{\varepsilon}
\def\bmo{\text{bmo}}

\numberwithin{equation}{section}
\theoremstyle{plain}                
\newtheorem{theorem}{Theorem}[section]

\newtheorem{proposition}[theorem]{Proposition}

\theoremstyle{definition}           
\newtheorem{definition}[theorem]{Definition}

\theoremstyle{remark}               
\newtheorem{remark}{Remark}[section]

\begin{document}

\pagenumbering{arabic} \pagestyle{plain}
\begin{center}
  {\large \bf Universal basic income in a financial equilibrium}
  \ \\ \ \\
Kim Weston\footnote{Department of Mathematics, Rutgers University, Piscataway, NJ, 08854 USA (kw552@rutgers.edu).}
\\

Rutgers University
\ \\


\today
\end{center}
\vskip .5in

\begin{abstract}
  Universal basic income (UBI) is a tax scheme that uniformly redistributes aggregate income amongst the entire population of an economy.  We prove the existence of an equilibrium in a model that implements universal basic income.  The economic agents choose the proportion of their time to work and earn wages that can be used towards consumption and investment in a financial market with a traded stock and annuity.  A proportion of the earned wages is uniformly distributed amongst all agents, leading to interconnectedness of the agents' decision problems, which are already dependent on one another through the financial market.  The decision problems are further entangled by perceptions of labor; the agents respond to the labor choices of others and act upon their perceived income in their decision problems.  The equilibrium is constructed and proven to exist using a backward stochastic differential equation (BSDE) approach for a BSDE system with a quadratic structure that decouples.  We analyze the effects of a universal basic income policy on labor market participation, the stock market, and welfare.  While universal basic income policies affect labor market participation and welfare monotonically, their effects on the stock market are nontrivial and nonmonotone.
\end{abstract}

\vskip -.15in
\noindent{\bf Acknowledgements:} The author would like to thank Frank Riedel and the anonymous referees for their constructive comments that greatly improved the paper.  While the original submission was done without any AI usage, the author used ChatGPT to proofread the mathematical arguments and obtain feedback on the exposition during the revision process.\\

\section{Introduction}\label{section:intro}
Universal basic income (UBI) policies are tax schemes that redistribute income amongst an entire population so that everyone receives an equal portion of aggregate income.  
Many local and regional governments have sought to alleviate poverty, encourage upward socioeconomic mobility, and dampen the effects of technological automation by means of UBI experiments\footnote{See for example, the works of Yang~\cite{Y18} and Hughes~\cite{H18}.}.  Over the last two decades, notable trials of UBI have played out across dozens of regions across the globe with a high concentration of experiments in the US.

While economists, sociologists, and policy makers have overseen and analyzed UBI experiments, current research is limited to understanding individual impacts and region-wide questions of budgeting; see, for example, West et al.~\cite{CDW23NMH}, Gertler et al.~\cite{GMR12AEJ}, and Kueng~\cite{K18QJE}.  UBI experiments suffer from three critical flaws:  they are expensive, they take many months or years to play out, and they are small in scale.  In particular, the expense leads to their small-scale nature, forcing UBI experiments to only directly benefit the sliver of low income individuals who win a lottery in order to receive basic income allotments.  Such experiments are not truly universal, and they cannot address broader financial economic concerns.  This work models universal basic income within a financial equilibrium in order to address questions that were previously unanswerable due to experimental limitations.

We propose a financial equilibrium model to study UBI with agents optimizing their labor market participation, running consumption, and investment in the stock market.  
The financial market consists of a stock with continuous stochastic dividends and an annuity with a continuous constant dividend stream of one.  The presence of a traded annuity builds off the financial equilibrium work of Christensen and Larsen~\cite{CL14RAPS}, Weston and \v{Z}itkovi\'c~\cite{WZ20FS}, and Weston~\cite{W24FS}.


The agents earn income by participating in the labor market with exogenously-specified wage rates.  The labor choices of others impact all agents in two ways.  First, they determine the aggregate income, which is partially redistributed to everyone as universal basic income.  Second, the agents respond to others' labor choices through response functions, similar to those originally proposed for trading responses in Vayanos~\cite{V99RES} and used subsequently in Chen et al.~\cite{CCLS23FS}.  Response functions provide the mechanism for the agents to perceive how others will respond to their labor efforts.  From perceived labor choices, the corresponding perceived income streams enter into each agent's investment and consumption choice problem.  Our response functions contain a free parameter, called the influence parameter, describing the influence of others on labor choices.  The size of the influence parameter impacts labor participation, the financial market, and welfare, showing its perceptions matter for the impact of UBI policies on equilibrium.

Labor-leisure problems study the trade-off between earning income by working in the labor market versus deriving pleasure from leisure time.  Bodie et al.~\cite{BMS92JEDC} studied the labor-leisure problem in a continuous-time single agent model, where the agent maximized expected utility from running consumption and leisure from not working.  Basak~\cite{B99JEDC} extended the work of Bodie et al.~\cite{BMS92JEDC}, placing labor-leisure problems in the context of a complete market financial equilibrium.  The setting of Basak~\cite{B99JEDC} is the most similar setting to this work, although it does not consider tax schemes like UBI or a response function component to the labor choice problem. Kenc~\cite{K04JEDC} studied the labor-leisure choice problem in a general equilibrium from the supply side and included tax scheme modeling.

The construction of an equilibrium relies on solving a characterizing BSDE system, similar to Weston and \v{Z}itkovi\'c~\cite{WZ20FS} and Weston~\cite{W24FS}.  We focus on a financial market with two traded securities and one source of Brownian risk, so even though the characterizing BSDE system is coupled and quadratic, the quadratic terms are not coupled.  This BSDE system is readily solvable using existing results.


We analyze the effects of universal basic income policies on labor participation, stock market effects, and welfare, and we compare pure capitalism and a social market economy in Section~\ref{section:effects}.  The effects depend on both the amount of redistribution of income and the influence parameter that describes the agents' labor perceptions.  For an influence parameter less than one, labor market participation decreases as income redistribution increases, but for influence parameter values greater than one, the opposite effect is achieved.  Welfare effects are studied for the case when wage rates are positive constants.  In this case, welfare behaves similarly to labor market participation; it is decreasing (increasing) when the amount of redistribution of income is increasing for small (large) influence parameter values.

The stock market effects are the most subtle.  While closed-form expressions are provided to describe the market price of risk and interest rate coefficients that appear in a state price deflator, they are nonmonotone and highly nontrivial.  An example is provided, highlighting some of the subtle variations in stock market effects due to UBI parameters.

\label{comment6}
While we focus on the financial economic and labor response aspects of UBI implementation from a theoretical academic perspective, Daruich and Fern\'andez~\cite{DF24AER} and Ghatak and Maniquet~\cite{GM19ARE} approach UBI implementation from a more practical economic perspective.  Daruich and Fern\'andez~\cite{DF24AER} model UBI in a stationary equilibrium with a representative cohort of agents that move through twenty stages (overlapping generations) of life, making decisions from getting a college education to saving for retirement to transferring wealth to children.  They find that UBI policies impact welfare negatively.  The negative effects are especially onerous on future generations because additional taxation limits future generations' skill development, education, and capital accumulation.  The negative welfare effects outweigh the positive effects, such as providing a safety net for adults displaced in a risky economy or automated job market.  Ghatak and Maniquet~\cite{GM19ARE} study UBI implementation from a theoretical normative perspective that is not equilibrium-based.  They conclude that UBI implementation is potentially reasonable only in two circumstances:  in developing countries where traditional welfare institutions are impediments to helping the poor, and as a short term complement to other conditional transfer policies.  Both \cite{GM19ARE} and \cite{DF24AER} provide rich, theoretical models for understanding the potential impacts of UBI policies in practice.  Our focus is different: we focus on the effects of a UBI policy on endogenously valued financial markets and use labor perceptions and trading in the financial market as the agent interaction mechanisms.

The structure of the paper is as follows.  Section~\ref{section:setting} sets up the model and states the main theorem, Theorem~\ref{theorem:main-1}.  Section~\ref{section:construction} constructs an equilibrium by defining response functions, defining the characterizing BSDE system, and postulating optimal strategies.  Theorem~\ref{thm:main} is a restatement of Theorem~\ref{theorem:main-1}, which states that an equilibrium exists and provides the details of its construction.  Section~\ref{section:effects} describes the effects of universal basic income on labor market participation, the financial market, and welfare.  
Section~\ref{section:proofs} provides the proofs.

\section{Model setting}\label{section:setting}

We study a continuous-time financial equilibrium in a pure exchange economy with a fixed finite time horizon $T<\infty$.  We let $B=(B_t)_{t\in[0,T]}$ be a one-dimensional Brownian motion on the probability space $(\Omega,\sF,\bP)$ equipped with the augmented Brownian filtration $\bF=(\sF_t)_{t\in[0,T]}$. We assume that $\sF_T = \sF$.  Throughout this work, equality (and inequality) between random variables is assumed to hold $\bP$-a.s., and we suppress time from the notation when possible.  

The set of adapted, continuous, uniformly bounded processes is denoted by
$\sS^\infty$. For $q>1$, $\sS^q$ denotes the set of adapted, continuous
processes $X$ such that
$$
    \E\left[\sup_{0\leq t\leq T}|X_t|^q\right]<\infty.
$$
A martingale $M$ is called a \emph{BMO-martingale} if there exists a
constant $C>0$ such that, for every stopping time $\tau\leq T$,
$$
    \E\left[
        \langle M\rangle_T-\langle M\rangle_\tau
        \,\middle|\,\sF_\tau
    \right]\leq C.
$$
In this case, we write $M\in\text{BMO}$. The collection of progressively measurable processes $\sigma$ for which $\int_0^\cdot\sigma\,dB$ is a BMO-martingale is denoted by $\bmo$.  The collection of progressively measurable processes $\eta$ such that
  $$
    \sup_{\tau\leq T} \left\|\E\left[\int_\tau^T|\eta_s|\,ds 
    \,\middle|\,\sF_\tau\right]\right\|_\infty < \infty
  $$
is denoted by $\bmo^{1/2}$. Background and details on BMO-martingales can be found in Kazamaki~\cite{Kaz94}.
\ \\

\noindent{\bf Economic agents.} 
The economy consists of $2\leq  I<\infty$ agents who have the choice between consumption of the real good, investment in the financial market, and working (or not) in the formal labor workforce in order to earn income. Each agent faces a decision problem to decide simultaneously how much to consume, invest, and work.  Consumption and labor occur at rates $c = (c_t)_{t\in[0,T]}$ and $L = (L_t)_{t\in[0,T]}$, respectively, per unit time and in a lump sum at time $T$.  

We model the agents' decision problems using Cobb-Douglas utility functions.  Each agent $i=1,\ldots, I$ has constant parameters for risk aversion $\alpha_i>0$, time preference $\rho_i\geq 0$, and labor/leisure preferences $u_i:(0,1)\rightarrow(0,\infty)$ that factor into the utility functions $U_i$ through
$$
  U_i(t, c, l) := -\exp\left(-\rho_i t-\alpha_i c\right)u_i(l),
   \quad t\in[0,T],\ c\in\R,\ l\in(0,1).
$$
The labor/leisure preferences are given by $u_i(l) = l^{\beta_i}(1-l)^{\gamma_i}$ for $l\in(0,1)$, where $\beta_i, \gamma_i < 0$ so that $u_i$ is strictly convex (so that $-u_i$ is strictly concave).  Under additional boundedness assumptions, it is possible to extend $u_i$ to the more typical case with $\beta_i = 0$, such as in Basak~\cite{B99JEDC}.  This extension is discussed below in Remark~\ref{rmk:ui}.  Consumption $c$ is allowed to be positive or negative, due to the exponential preferences for consumption.  The proportion of time spent on labor must be valued in $(0,1)$.  For every labor proportion $l\in(0,1)$, the proportion of time spent on leisure is $1-l\in(0,1)$.  

When an agent works in the labor market, she earns a wage that contributes to her total wealth.  However, some of her wage is taxed, which is centrally collected and uniformly redistributed to everyone.  The uniform redistribution of taxed funds is what we refer to as {\it universal basic income}.  It is universal because everyone receives it and basic because nobody needs to work in the labor market in order to receive it.  Every agent will keep the proportion $\lambda_{keep}\in[0,1]$ of their earned labor income, and every agent will receive the proportion $\frac{\lambda_{ubi}}{I}$ of the aggregate income, where $\lambda_{ubi}\in[0, 1-\lambda_{keep}]$.  We note that $\lambda_{keep}+\lambda_{ubi}\leq 1$.  If $\lambda_{keep}+\lambda_{ubi}<1$, then some of the economy's earned labor income is lost, presumably as frictions or administrative costs that are not recycled back into the economy.

For each $i=1,\ldots, I$, agent $i$ has the ability to earn the wage $w_i = (w_{i,t})_{t\in[0,T]}$ per unit time and in a lump sum at time $T$.  The wage rate $w_i$ has dynamics
$$
  dw_{i} = \mu_{w_i} dt + \sigma_{w_i} dB, 
$$
where $\mu_{w_i}, \sigma_{w_i}\in\sS^\infty$ are progressively measurable and uniformly bounded.  The wage rate $w_i$ is assumed to be nonzero $dt\otimes \bP$-a.e. and $\bP$-a.s. at time $T$.
 
  For each $j=1,\ldots, I$, if we denote agent $j$'s labor proportion as $\widetilde L_j = (\widetilde L_{j,t})_{t\in[0,T]}$, then agent $i$'s income rate at time $t\in[0,T]$ is\label{comment8}
\begin{equation*}\label{eqn:income-formula}
  \lambda_{keep} w_{i,t} \widetilde L_{i,t} + \frac{\lambda_{ubi}}{I}\sum_{j=1}^I w_{j,t} \widetilde L_{j,t}.
\end{equation*}
\ \\

\noindent{\bf Labor response functions.} We employ an equilibrium concept in which markets must clear and the agents will act optimally, but the agents will also take into consideration how they perceive others will respond to their labor choices when making their own choices.  Each agent earns income by working, but a portion of that income is siphoned off and redistributed to everybody as UBI.  Though the agents benefit from their own additional earnings, they also benefit from the labor earnings of others.  Thus, others' labor choices should have an impact on labor decisions.  We model this impact by using affine response functions, similar to Vayanos~\cite{V99RES} and Chen et al.~\cite{CCLS23FS}.  Response functions provide one way of making others' impact precise.\label{comment5}

For each $i, j=1, \ldots, I$ and $j\neq i$, the response function $\Lambda^i_j=\Lambda^i_{j,t}(l, \omega)$ describes how agent $i$ perceives that agent $j$'s labor choice changes according to agent $i$'s choice of labor $l\in (0,1)$ at time $t\in[0,T]$ for $\omega\in\Omega$.  We assume that $\Lambda^i_j : \Omega\times[0,T]\times(0,1)\rightarrow\R$ is $\mathcal{P}(\bF)\otimes\mathcal{B}((0,1))$-measurable.  Then, agent $i$'s perceived income rate,\label{comment2} based on agent $i$'s choice of labor for $l\in(0,1)$, is given by $\eps_i=\eps_i(l)$,
\begin{align}\label{def:perceived-income}
  \eps_{i,t}(l):= \lambda_{keep} w_{i,t} l &+ \frac{\lambda_{ubi}}{I}\left(w_{i,t} l + \sum_{j\neq i} w_{j,t}\Lambda^i_{j,t}(l) \right), \quad t\in[0,T].
\end{align}
We suppress the $\omega$ from the notation in $\Lambda^i_{j,t}(l)$.  Below, in Definition~\ref{def:response-fn}, we consider affine responses.
\ \\

\noindent{\bf The financial market.} 
The financial market consists of two securities:  a stock and an annuity.  Both securities are in one-net supply, and their prices are denominated in units of a single consumption good.  The annuity pays a constant dividend stream of $1$ per unit time and a lump sum of $1$ at time $T$.  The stock pays a dividend of $D_t$ per unit time and a lump sum of $D_T$ at time $T$.  The dividend stream $D$ is an It\^o process with dynamics given by
\begin{equation}\label{def:state-process}
  dD_t = \mu_{D}dt + \sigma_{D}dB, \quad D_0\in\R,
\end{equation}
where $\mu_D, \sigma_D\in\sS^\infty$.

The stock and annuity prices will be determined endogenously in equilibrium as continuous semimartingales $S=(S_t)_{t\in[0,T]}$ and $A=(A_t)_{t\in[0,T]}$, respectively.  Their terminal values are their dividends
$$
  S_T = D_T \quad \text{and} \quad A_T = 1.
$$
Our equilibrium prices $A$ and $S$ have dynamics of the form
\begin{align}
  dA &= -dt + A\left(\mu_{A}\,dt+\sigma_{A}\,dB\right), \quad A_T=1,
  \label{def:A-dynamics} \\
  dS &= -D_t\,dt + \left(\mu_{S}+S\mu_{A}\right)dt + \left(\sigma_{S}+S\sigma_{A}\right)dB, \quad S_T = D_T.\label{def:S-dynamics}
\end{align}
The price processes and their dynamics are outputs of equilibrium.  The coefficients $\mu_A$, $\sigma_A$, $\mu_S$, and $\sigma_S$ must be progressively measurable and will be proven to have sufficient regularity for \eqref{def:A-dynamics} and \eqref{def:S-dynamics} to be well-defined.

The annuity plays the role of a zero-coupon bond when a zero-coupon bond trades in an economy with only terminal consumption.  When the annuity's volatility is zero, it can replicate a locally riskless security.

The number of shares held in the stock over time is denoted $\pi = (\pi_t)_{t\in[0,T]}$, and the number of shares held in the annuity over time is denoted $\theta = (\theta_t)_{t\in[0,T]}$.  For given price processes $S$ and $A$ and a pair of positions $(\pi,\theta)$, the associated wealth process $X = (X_t)_{t\in[0,T]}$ is defined by
$$
  X := \theta A + \pi S.
$$

The agents can trade in both the stock and the annuity and are endowed with $\pi_{i,0-} \in\R$ shares in the stock and $\theta_{i,0-}\in\R$ shares in the annuity, $i=1,\ldots, I$.  Since the stock and annuity are in one-net supply, we assume that $\sum_{i=1}^I \theta_{i,0-} = 1$ and $\sum_{i=1}^I \pi_{i,0-}=1$.

\begin{definition}\label{def:adm} Let $i\in\{1,\ldots,I\}$ be given.  Progressively measurable strategies $\pi$, $\theta$, $c$, and $L$ are {\it admissible for agent $i$} with the associated wealth process $X:=
\theta A + \pi S$ if
\begin{enumerate}
  \item \label{adm:1} $L$ is $(0,1)$-valued and
  \begin{align*}
    \int_0^T &\big(|X_t\mu_{A,t}| + |\pi_t\mu_{S,t}| + |\eps_{i,t}(L_t)| + |c_t| 
    + |X_t\sigma_{A,t}|^2 + |\pi_t\sigma_{S,t}|^2\big)dt <\infty; 
  \end{align*}

  \item \label{adm:dyn} The associated wealth process has initial wealth $X_0 = \pi_{i,0-}S_0 + \theta_{i,0-}A_0$ and satisfies the self-financing condition:
    \begin{align}\label{eqn:self-financing}
      dX = \big(X \mu_A + \pi \mu_S + \eps_i(L) - c\big)dt + \big(X\sigma_A + \pi\sigma_S\big)dB;
    \end{align}
  
  \item\label{adm:3} We require that $\sM^i(\pi,\theta,c,L)$ is a supermartingale, where $\xi_{i,t} := \tfrac{\partial}{\partial c} U_i (t, c_{i,t}, L_{i,t})$ and
    $$
      \sM^i_t(\pi,\theta,c,L) := \xi_{i,t}\left(X_t-X_{i,t}\right)
      + \int_0^t \xi_{i,s}\big((c_s-\eps_{i,s}(L_s))-(c_{i,s}-\eps_{i,s}(L_{i,s}))\big)ds.
    $$
    Here, $c_i$, $L_i$, and $X_i$ are defined below in \eqref{def:ci}, Definition~\ref{def:Li}, \eqref{def:Xi}, respectively.  The terms $c_i$, $L_i$, and $X_i$ are defined only in terms of inputs, with $c_i$ and $X_i$ depending on the BSDE \eqref{def:bsde}, whose existence and uniqueness is established in Proposition~\ref{prop:bsde} in terms of only inputs.  
\end{enumerate}
If $(\pi,\theta, c, L)$ is admissible for agent $i$, then we write $(\pi,\theta, c,L)\in\sA_i$.
\end{definition}
Item~\ref{adm:3} of Definition~\ref{def:adm} introduces an agent-specific dual process $\xi_i$ for use in the verification proof of Theorem~\ref{thm:main}. Because security price volatilities are determined in equilibrium, market completeness is itself an endogenous property. The dynamic spanning argument originates with Duffie and Huang~\cite{DH85E}, while sufficient conditions for endogenous completeness in continuous-time equilibrium models are developed by Anderson and Raimondo~\cite{AR08E}, Riedel and
Herzberg~\cite{RH13JME}, Hugonnier, Malamud, and Trubowitz~\cite{HMT12E}, and Kramkov and Predoiu~\cite{KP14SPA}. In our one-dimensional Brownian setting, the annuity-discounted stock price $S/A$ has volatility $\sigma_S/A$, so the market is endogenously complete when $\sigma_S\neq0$ $dt\otimes d\mathbb P$-a.e. In this case, each $\xi_i$ is a state-price deflator, and the normalized deflators $\xi_i/\xi_{i,0}$ are identical across agents. If $\sigma_S$ vanishes on a set of positive $dt\otimes d\mathbb P$ measure, the dynamic spanning condition fails, and the admissibility condition above allows the verification argument to proceed without imposing
endogenous completeness.

\begin{definition}\label{def:eq}
  Let $\lambda_{keep}\in[0,1]$, $\lambda_{ubi}\in[0,1-\lambda_{keep}]$ be given. Strategies $\pi_i, \theta_i, c_i, L_i$, $i=1,\ldots, I$, response functions $\Lambda^i_j$, $i,j=1,\ldots, I$ with $i\neq j$, and continuous semimartingales $A$ and $S$ form an {\it equilibrium} if
\begin{enumerate}
  \item {\it Optimality:} For each $i=1,\ldots, I$, we have that $\pi_i$, $\theta_i$, $c_i$, and $L_i$ solve
  \begin{align}\label{def:optimization}
    \sup_{\pi, \theta, c, L} \E\left[\int_0^T U_i(t,c_t,L_t) dt + U_i\big(T,X_T+\eps_{i,T}(L_T), L_T\big)\right],
  \end{align}
  where the supremum is taken over $(\pi, \theta, c,L)\in\sA_i$.  Here, the perceived income rate process $\eps_i(L)$ is given by \eqref{def:perceived-income}.

  \item {\it Consistency of optimizers:} For all $t\in[0,T]$ and each $i, j = 1,\ldots, I$ with $i\neq j$, we have $\Lambda^i_{j,t}(L_{i,t}) = L_{j,t}$. 
  
  \item {\it On-equilibrium perceptions align with reality:} \label{eq:perceptions}
  For all $t\in[0,T]$,
  $$
    \sum_{i=1}^I \eps_{i,t}(L_{i,t}) = \left(\lambda_{keep}+\lambda_{ubi}\right)\sum_{i=1}^I w_{i,t} L_{i,t}.
  $$
  \item {\it Market clearing:} For all $t\in[0,T]$, we have
  \begin{align*}
    \sum_{i=1}^I \pi_{i,t} = 1,\quad
    \sum_{i=1}^I \theta_{i,t} = 1, \quad \text{and} \quad
    \sum_{i=1}^I c_{i,t}= 1+D_t+\sum_{i=1}^I \eps_{i,t}(L_{i,t}).
  \end{align*}
\end{enumerate}

\end{definition}
The definition of equilibrium incorporates optimally choosing a proportion of time to work in the labor market with optimizing holdings in the financial market and running consumption, where income earned via wages from the labor market is treated as income available in the financial market.  In the optimization problem \eqref{def:optimization}, requiring $(\pi, \theta, c, L)\in\sA_i$ means that the income stream being optimized over is $\eps_i(L)$.

The labor perceptions $\Lambda_j^i(l)$ are agent $i$'s perception of $j$'s reaction to $i$'s choice of labor.  The consistency requirement $\Lambda^i_j(L_i)=L_j$ means that agent $i$ perceives that $j$ responds to $i$'s optimal labor choice with his own optimal labor choice.  The consistency requirement also guarantees that the on-equilibrium labor choice $\Lambda_j^i(L_i)$ is $(0,1)$-valued, whereas there is no such requirement off-equilibrium.  \label{comment1}We note that the optimal labor proportions $L_1, \ldots, L_I$ and the response functions $\Lambda^i_j$ are {\it outputs} of equilibrium, not inputs.  Candidate labor proportions and candidate response functions will be defined in Definitions~\ref{def:Li} and \ref{def:response-fn}, respectively.  Although we construct an equilibrium in Section~\ref{section:construction} using Definitions \ref{def:Li} and \ref{def:response-fn}, it does not preclude other possible labor proportions and response functions.

Our main result is Theorem~\ref{theorem:main-1}, which establishes the existence of an equilibrium.  The equilibrium construction is contained in Section~\ref{section:construction}, culminating in Theorem~\ref{thm:main}, which implements Theorem~\ref{theorem:main-1}.  The proof is in Section~\ref{section:proofs}.
\begin{theorem}\label{theorem:main-1}
  Assume that the wage rates are nonzero $dt\otimes\bP$-a.e. and $\bP$-a.s. at time $T$, and that for the dividend and wage rate dynamics coefficients, we have $\mu_D, \sigma_D, \mu_{w_i}, \sigma_{w_i} \in\sS^\infty$ for $i=1,\ldots, I$.  For any $\lambda_{keep}\in[0,1]$ and $\lambda_{ubi}\in[0,1-\lambda_{keep}]$, there exists an equilibrium.
\end{theorem}

\section{Equilibrium construction}\label{section:construction}

To construct an equilibrium, we begin by making two conjectures.  First, we conjecture the agents' optimal labor proportions using only model input parameters.  Second, we conjecture that affine response functions are sufficient for constructing and proving the existence of an equilibrium.  Given the candidate optimal labor proportions and their corresponding income streams, we  construct the equilibrium financial market.  Finally, using these conjectures, we prove in Section~\ref{section:proofs} that equilibrium exists by solving a characterizing BSDE system and performing verification.\\

\noindent{\bf The labor market.}  The candidate labor proportion processes $L_1, \ldots, L_I$ are defined below and depend only on input parameters and a constant free parameter $\delta\in\R$, called the influence parameter.
\begin{definition}\label{def:Li}
  For $\delta\in\R$ and for each $i=1, \ldots, I$, the process $L_i=(L_{i,t})_{t\in[0,T]}$ is chosen for each $t\in[0,T]$ as the unique value in $(0,1)$ such that 
  $$
    \alpha_i w_{i,t} \left(\lambda_{keep}+\frac{\lambda_{ubi}}{I}\big(1+\delta(I-1)\big)\right) = \frac{u_i'(L_{i,t})}{u_i(L_{i,t})}.
  $$
\end{definition}

For each $i=1,\ldots, I$, the process $L_i$ is an It\^o process with dynamics
$$
  d L_i = \mu_{L_i} dt + \sigma_{L_i} dB.
$$
Furthermore, we denote the processes $\sL_i$ and $\sL_\Sigma$ by $\sL_i := \frac{1}{\alpha_i} \log u_i(L_i)$ and $\sL_\Sigma := \sum_{i=1}^I \sL_i$.  $\sL_i$ and $\sL_\Sigma$ are It\^o processes with dynamics denoted by
$$
  d\sL_i = \mu_{\sL_i} dt + \sigma_{\sL_i} dB
  \quad \text{and} \quad
  d\sL_\Sigma = \mu_\sL dt + \sigma_\sL dB.
$$

Next, we define candidate labor response functions, which are affine in $l$.
\begin{definition}\label{def:response-fn}
  For each $i, j=1,\ldots, I$ with $j\neq i$ and $\delta\in\R$, $\Lambda_j^i: \Omega\times[0,T]\times(0,1)\rightarrow \R$ is defined by
  $$
    \Lambda_{j,t}^i(l) := \delta\frac{w_{i,t}}{w_{j,t}}(l-L_{i,t}) + L_{j,t}, \quad l\in(0,1),
  $$
  where $\omega\in\Omega$ is suppressed from the notation, and $L_1, \ldots, L_I$ are the candidate optimal labor proportions defined in Definition \ref{def:Li}.
\end{definition}
We recall that the wage rates are nonzero $dt\otimes\bP$-a.e. and $\bP$-a.s. for $t=T$, and we define $\Lambda^i_{j,t}(l) = L_{j,t}$ when $w_{j,t}=0$.  For each $t\in[0,T]$, $\Lambda^i_{j,t}$ is affine in $l$.  The response function $\Lambda^i_{j,t}(l)$ describes how agent $i$ perceives agent $j$ to respond to $i$'s choice of labor $l$.  The influence parameter $\delta$ is chosen to be the same across all agents, and the sign of $\delta$ describes the overall attitude of perceptions of the economic agents.  We call $\delta>0$ the {\it greater good} scenario because the agents perceive that the other agents respond positively to their hard work (labor).  Everybody's perceived income increases in the greater good scenario.  Conversely, we call the $\delta<0$ case the {\it freeloader} scenario because the agents perceive a negative response from others, and everybody's perceived income suffers as a result.  Finally, $\delta = 0$ is the {\it competitive} scenario, in which the agents perceive no response to the labor choices of others.

For a fixed $i=1,\ldots, I$, using the labor response functions $\Lambda^i_{j,t}$ in Definition~\ref{def:response-fn} and by \eqref{def:perceived-income}, agent $i$'s perceived income for $l\in(0,1)$ and $t\in[0,T]$ is
\begin{align}\label{epsi}
  \eps_{i,t}(l) = w_{i,t} \left(\lambda_{keep}+\frac{\lambda_{ubi}}{I}\big(1+\delta(I-1)\big)\right) l + \frac{\lambda_{ubi}}{I}\left(\sum_{j\neq i} w_{j,t}L_{j,t} - \delta(I-1)w_{i,t}L_{i,t}\right),
\end{align}
$dt\otimes\bP$-a.e. and $\bP$-a.s. at time $T$.

Based on the processes $L_i$ from Definition~\ref{def:Li} for agent $i$'s candidate optimal labor proportion, we denote the candidate aggregate income stream by $\eps_\Sigma:=\sum_{i=1}^I \eps_{i}(L_{i})$.  Since $\eps_i(L_i)$ and $\eps_\Sigma$ are It\^o processes, we write their dynamics as
$$
  d\eps_i(L_i) = \mu_{\eps_i} dt + \sigma_{\eps_i} dB
  \quad \text{and} \quad
  d\eps_{\Sigma} = \mu_{\eps} dt + \sigma_{\eps} dB.
$$

Proposition \ref{prop:Li} establishes uniform boundedness of several drift and volatility terms related to labor and income, which will be needed for the construction and verification of equilibrium.  The proof of Proposition~\ref{prop:Li} is in Section~\ref{section:proofs} below.
\begin{proposition}\label{prop:Li} 
For each $i=1,\ldots, I$, assume that the wage rates are nonzero $dt\otimes\bP$-a.e. and $\bP$-a.s. at time $T$, and that for the wage rate dynamics coefficients, we have $\mu_{w_i}, \sigma_{w_i} \in\sS^\infty$. The progressively measurable processes $\mu_{L_i}$, $\mu_{\sL_i}$, $\mu_{\eps_i}$, $\sigma_{L_i}$, $\sigma_{\sL_i}$, and $\sigma_{\eps_i}$ are uniformly bounded.  Moreover, the drift and volatility dynamics coefficients of $L_i w_i$ are also uniformly bounded.
\end{proposition}

We construct a financial market and determine the agents' consumption and investment strategies in terms of the optimal labor proportion and income streams with the help of the solution $\big((a, Z_a), (Y_i, Z_i)_{1\leq i\leq I}\big)$ to the BSDE system
\begin{align}
  \begin{split}\label{def:bsde}
    da &= Z_a\, dB + \left(\rho_\Sigma +\alpha_\Sigma(\mu_D+\mu_\eps-\mu_\sL) -\frac{1}{2}\big(\alpha_\Sigma(\sigma_D+\sigma_\eps-\sigma_\sL)-Z_a\big)^2- \exp(-a)\right)dt,\\
    dY_i &= Z_i\,dB + \frac{1}{\alpha_i}\left(-\rho_i + \frac{1+a+\alpha_i(Y_i-\eps_i(L_i)+\sL_i)}{\exp(a)}-\frac{1}{2}\big(\alpha_\Sigma(\sigma_D+\sigma_\eps-\sigma_\sL)-Z_a\big)^2\right.\\
    &\quad\quad\quad\quad\quad\quad\left.\phantom{\frac{(Y_i)}{A}}+\alpha_iZ_i\big(\alpha_\Sigma(\sigma_D+\sigma_\eps-\sigma_\sL)-Z_a\big)\right)dt,\\
  a_T &= 0, \quad Y_{i,T}= \eps_{i,T}(L_{i,T})-\sL_{i,T}, \quad 1\leq i\leq I,
  \end{split}
\end{align}
where $\alpha_\Sigma := \left(\sum_{i=1}^I \frac{1}{\alpha_i}\right)^{-1}$ and $\rho_\Sigma:=\alpha_\Sigma\sum_{i=1}^I \frac{\rho_i}{\alpha_i}$ are constants.  
The  BSDE system \eqref{def:bsde} is $(I+1)$-dimensional, quadratic, and coupled.  However, the $(a, Z_a)$ equation decouples from the system, and the coupling is not quadratic.  \label{comment3}The processes $Y_i$ represent agent $i$'s conditional certainty equivalent value.  At $t$, $Y_{i,t}$ is conditionally how much agent $i$ would require at time $t$ to achieve the equivalent expected utility to the optimal consumption and investment rates.

\begin{remark}\label{comment4}
  The BSDE is affected by the influence parameter $\delta$ through several of its input parameters.  The influence parameter is introduced through $L_i$ in Definition~\ref{def:Li} and $\Lambda^i_j$ in Definition~\ref{def:response-fn}.  Therefore, it affects $\eps_i(L_i)$, $\eps_\Sigma$, $\mu_\eps$, $\sigma_\eps$, $\sL_i$, $\sL_\Sigma$, $\mu_\sL$, and $\sigma_\sL$.  Unless the wage rates are all deterministic, $\delta$ will impact every BSDE output.
\end{remark}

Proposition~\ref{prop:bsde} establishes the existence and uniqueness of a solution to BSDE system \eqref{def:bsde}, which will pave the way for the financial market's equilibrium construction.  The proof is provided below in Section~\ref{section:proofs}.  
\begin{proposition}\label{prop:bsde}
  Assume that the wage rates are nonzero $dt\otimes\bP$-a.e. and $\bP$-a.s. at time $T$, and that for the dividend and wage rate dynamics coefficients, we have $\mu_D, \sigma_D, \mu_{w_i}, \sigma_{w_i} \in\sS^\infty$ for $i=1,\ldots, I$.  Then, there exists a unique solution $\big((a, Z_a), (Y_i, Z_i)_{1\leq i\leq I}\big)$ to the BSDE system \eqref{def:bsde} with $(a,Z_a)\in\sS^\infty\times\bmo$ and $Y_i - \eps_i(L_i) + \sL_i\in\sS^\infty$ and $Z_i\in\bmo$.
  
\end{proposition}

Let $\big((a, Z_a), (Y_i, Z_i)_{1\leq i\leq I}\big)$ be the unique solution to the BSDE system \eqref{def:bsde} guaranteed by Proposition \ref{prop:bsde}.  Using this solution, we construct our remaining candidate equilibrium quantities:  the stock price, the annuity price, optimal investment policies, optimal wealth processes, and optimal consumption policies. First, we denote $Y_\Sigma:=\sum_{i=1}^I Y_i$ and $Z_\Sigma:=\sum_{i=1}^I Z_i$ and observe that $(Y_\Sigma, Z_\Sigma)$ is the unique solution to the BSDE
\begin{align*}
  dY_\Sigma &= Z_\Sigma dB + \frac{1}{\alpha_\Sigma}\left(-\rho_\Sigma + \frac{1+a+\alpha_\Sigma(Y_\Sigma-\eps_\Sigma+\sL_\Sigma)}{A}-\frac{\alpha_\Sigma^2}{2}\left(\sigma_D+\sigma_\eps-\sigma_\sL\right)^2\right.\\
    &\quad\quad\quad\quad\quad\quad\left.\phantom{\frac{Y_\Sigma}{A}}+\alpha_\Sigma^2 Z_\Sigma\left(\sigma_D+\sigma_\eps-\sigma_\sL-\frac{Z_a}{\alpha_\Sigma}\right)\right)dt,
\end{align*}
with $Y_{\Sigma,T} = \eps_{\Sigma,T}-\sL_{\Sigma, T}$, where $Y_\Sigma - \eps_\Sigma + \sL_\Sigma \in \sS^\infty$ and $Z_\Sigma \in \bmo$.

Next, we define the annuity and stock price processes by
\begin{equation}\label{def:prices}
  A:=\exp(a) \quad \text{and} \quad S:= A\left(D+\eps_\Sigma -\frac{a}{\alpha_\Sigma}- Y_\Sigma-\sL_\Sigma\right).
\end{equation}
$A$ and $S$ satisfy the dynamics conjectured in \eqref{def:A-dynamics} and \eqref{def:S-dynamics} with $\mu_A$, $\sigma_A$, $\mu_S$, $\sigma_S$, and $\kappa$ defined by
\begin{align}
\begin{split}\label{eqn:terms}
  \kappa &:= \alpha_\Sigma(\sigma_D + \sigma_\eps - \sigma_\sL),\\
  \sigma_A &:= Z_a,\\
  \mu_A 
      &:= \rho_\Sigma + \alpha_\Sigma (\mu_D+\mu_\eps-\mu_\sL)-\frac{\kappa^2}{2}+\kappa \sigma_A,\\
  \sigma_S &:= \frac{A}{\alpha_\Sigma}\left(\kappa- \sigma_A -\alpha_\Sigma Z_\Sigma\right),\\
  \mu_S &:= \kappa \sigma_S.
  \end{split}
\end{align}
The process $\kappa$ can be interpreted as the market price of risk and is discussed below in Section~\ref{section:stock-market-effects}.

Theorem~\ref{thm:main} is the main result of this work and is the implementation of Theorem~\ref{theorem:main-1}.  The proof of Theorem~\ref{thm:main} is contained in Section~\ref{section:proofs}.
\begin{theorem}\label{thm:main}
  Assume that the wage rates are nonzero $dt\otimes\bP$-a.e. and $\bP$-a.s. at time $T$, and that for the dividend and wage rate dynamics coefficients, we have $\mu_D, \sigma_D, \mu_{w_i}, \sigma_{w_i} \in\sS^\infty$ for $i=1,\ldots, I$. Let $\lambda_{keep}\in[0,1]$, $\lambda_{ubi}\in[0,1-\lambda_{keep}]$, and $\delta\in\R$ be given.  Let $\big((a,Z_a), (Y_i, Z_i)_{1\leq i\leq I}\big)$ be the unique solution to the BSDE system \eqref{def:bsde} guaranteed by Proposition~\ref{prop:bsde}.  There exists an equilibrium with annuity and stock price processes given by \eqref{def:prices} and labor response functions given by Definition~\ref{def:response-fn} using the parameter $\delta\in\R$.  For $\kappa$, $\sigma_A$, $\mu_A$, $\sigma_S$, and $\mu_S$ given in \eqref{eqn:terms}, $A$ and $S$ satisfy the dynamics conjectured in \eqref{def:A-dynamics} and \eqref{def:S-dynamics}.  The agents' optimal labor proportions are the processes $L_1, \ldots, L_I$ defined in Definition~\ref{def:Li}. For $i=1,\ldots, I$, agent $i$'s optimal number of stock shares $\pi_i$ is
  \begin{equation}\label{def:pii}
    \pi_i = 
    \begin{cases}
      \frac{A}{\alpha_i\sigma_S}\left(\kappa - \sigma_A -\alpha_i Z_i\right), & \text{if }\sigma_S\neq 0\\
      \pi_{i,0-}, & \text{otherwise}
    \end{cases}.
  \end{equation}
  Agent $i$'s optimal wealth process is given by $X_{i,0-} =X_{i,0}= \theta_{i,0-}A_0 + \pi_{i,0-}S_0$ and for $t\in[0,T]$ by
  \begin{align}\label{def:Xi}
  \begin{split}
    X_{i,t} &= A_t\left(\tfrac{X_{i,0-}}{A_0} + \int_0^t \left(\tfrac{\pi_{i,s}}{A_s}(\mu_{S,s}-\sigma_{A,s}\sigma_{S,s}) - \tfrac{a_s+\alpha_i(Y_{i,s} - \eps_{i,s}(L_{i,s})+\sL_{i,s})}{\alpha_i A_s}\right)ds  + \int_0^t \tfrac{\pi_{i,s}\sigma_{S,s}}{A_s} dB_s\right).
  \end{split}
  \end{align}
  Agent $i$'s optimal consumption rate is
  \begin{equation}\label{def:ci}
    c_i = \frac{X_i}{A} + \frac{a}{\alpha_i}+Y_i + \sL_i,
  \end{equation}
  and the optimal annuity holdings are determined by $\theta_i = \frac{X_i-\pi_i S}{A}$.  Agent $i$'s quadruple of strategies is admissible with $(\pi_i, \theta_i, c_i,L_i)\in\sA_i$ and solves \eqref{def:optimization}.
\end{theorem}

\begin{remark}[Extension of the Labor/Leisure Utility Function, $u_i$]\label{rmk:ui} For each $i=1,\ldots, I$, we have assumed that the labor/leisure utility function is given by $u_i(l) = l^{\beta_i}(1-l)^{\gamma_i}$, $l\in(0,1)$, for $\beta_i, \gamma_i<0$ in \eqref{def:optimization}.  This form of $u_i$ has the nonstandard feature that more leisure is not always better.  While some empirical evidence suggests that working provides utility beyond earning wages, our labor utility function is nonstandard because it assigns $-\infty$ utility to both full labor participation and full leisure; see, for example, Young \cite{Y12SF}.

Under the additional assumptions that $\lambda:=\lambda_{keep}+\tfrac{\lambda_{ubi}}{I}(1+\delta(I-1))>0$ and $w_{i,t} > \tfrac{-\gamma_i}{\alpha_i\lambda}$ for each $i$ and $t\in[0,T]$, the conclusions of Theorem \ref{thm:main}, Proposition \ref{prop:Li}, and Proposition \ref{prop:bsde} hold for
$$
  u_i(l) = (1-l)^{\gamma_i}, \quad l\in(0,1),
$$
where $\gamma_i<0$ (so that $-u_i$ is strictly concave and strictly increasing in $(1-l)$).

This extension relies on two key steps.  First, the first-order condition for the labor problem must be satisfied, yielding optimal strategies $L_i$, as in Definition \ref{def:Li}.  Second, the drift and volatility dynamics terms in $L_i$, $\sL_i$, and $w_iL_i$ must remain uniformly bounded.  Under the stated positivity and lower-bound assumptions, the proofs of these two key steps are nearly identical to the current setting.

To visualize the extended form of $u_i$, we plot both the labor utility as a function of leisure in both the original and extended forms.  In the original formulation of $u_i$, the difference in order of magnitude between $\beta_i$ and $\gamma_i$ is significant.  Since the maximizer of the labor-preference component $-u_i$ is $\frac{\beta_i}{\gamma_i+\beta_i}$, the preferred labor proportion in the absence of income effects approaches zero as $\beta_i$ approaches zero.  Figure~1 
 plots labor utility as a function of leisure for two contrasting parameter choices:  $|\beta_i| \ll |\gamma_i|$ with $\beta_i = -0.01$, $\gamma_i=-1$ and $\beta_i = \gamma_i=-1$.
\begin{figure}[t]\label{figure:u}
    \begin{center}
    $
        \begin{array}{cc}
        \includegraphics[width=0.47\textwidth]{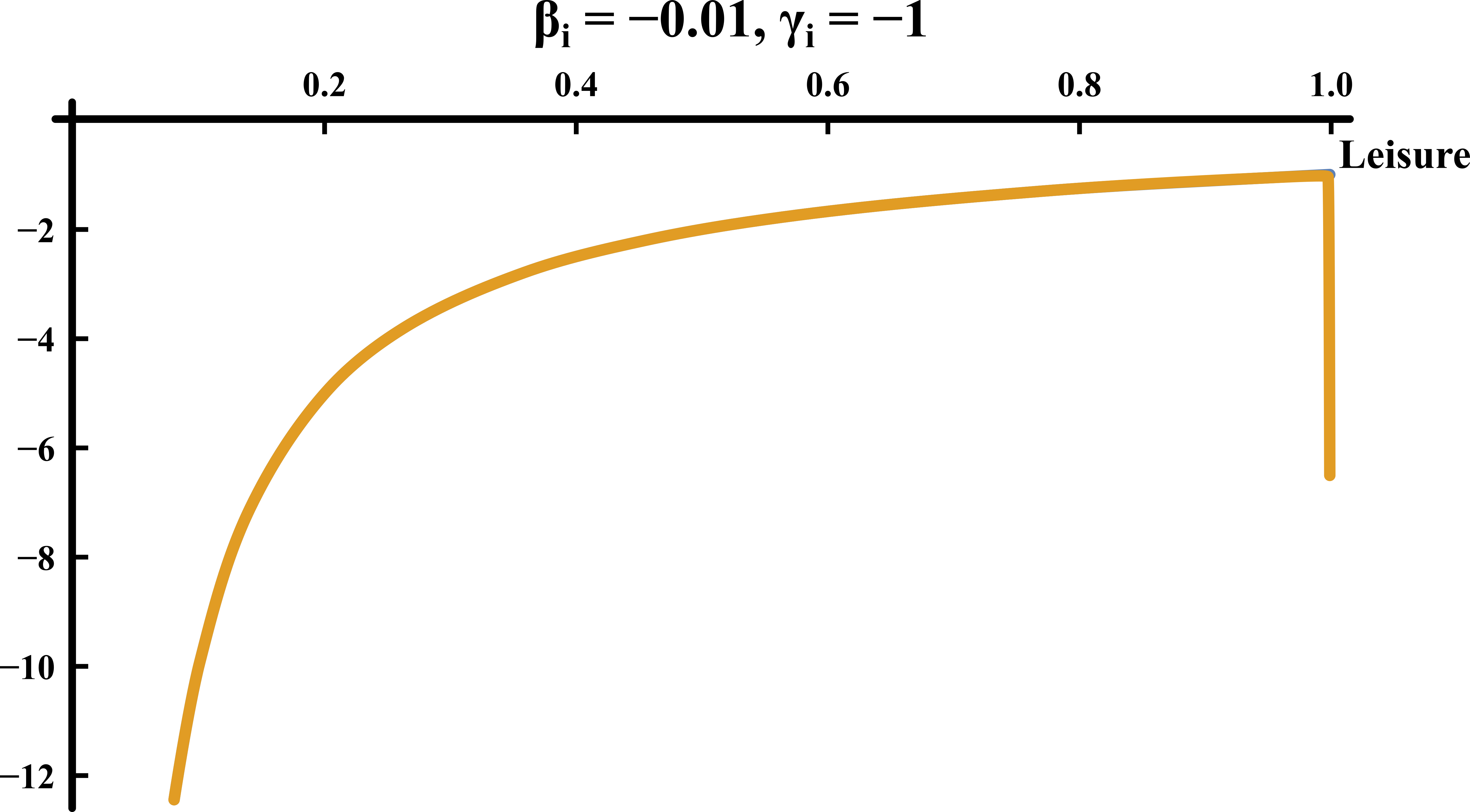} &
        \includegraphics[width=0.47\textwidth]{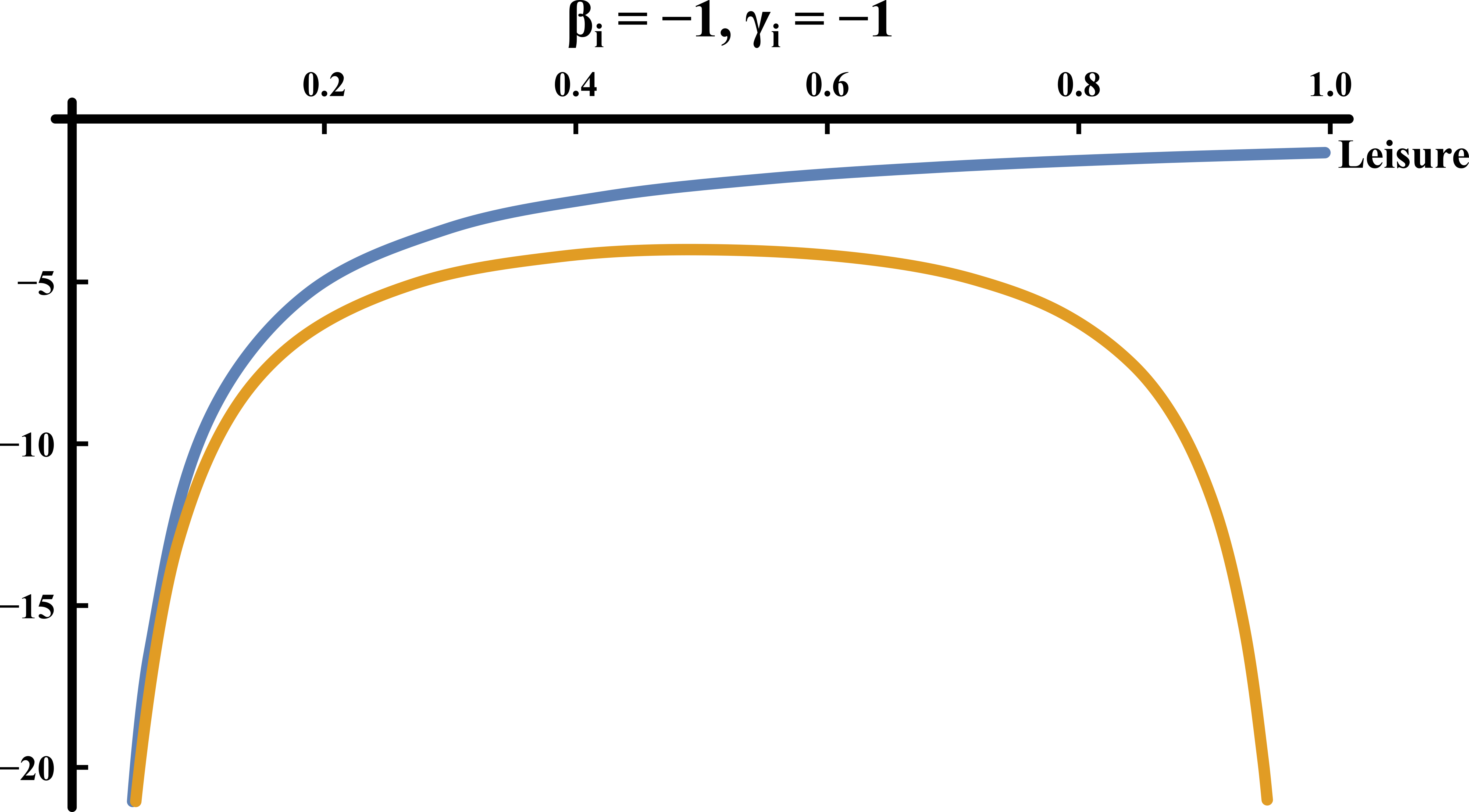}
        \end{array}
    $
    \end{center}
    \caption{Graphs of the labor-preference components as functions of leisure $\ell\in(0,1)$. The orange curves plot the original specification $-u_i(1-\ell)=-(1-\ell)^{\beta_i}\ell^{\gamma_i}$, and the blue curves plot the extended specification $-\ell^{\gamma_i}$. The left panel uses $\beta_i=-0.01$ and $\gamma_i=-1$, while the right panel uses $\beta_i=\gamma_i=-1$.}
    \label{fig:utility-shapes}
\end{figure}

\end{remark}

\section{Universal Basic Income Effects}\label{section:effects}
We study the impact of a UBI policy on the economy in terms of labor market participation, stock market effects, and welfare. {\it Throughout this section, we assume that all wage rate processes $w_i = (w_{i,t})_{t\in[0,T]}$, $i=1,\ldots,I$ are positive.}  Although the construction and existence of a universal basic income equilibrium does not rely on positivity of wage rates, drawing sensible and applicable conclusions about such equilibria does require wage rates to be positive.

\subsection{Labor Market Participation}\label{section:labor-mkt}

By Definition \ref{def:Li}, for each $i=1,\ldots, I$, $L_i$ is $(0,1)$-valued and determined by
$$
  \Psi_i(L_i) = \alpha_i w_i \left(\lambda_{keep}+\frac{\lambda_{ubi}}{I}\big(1+\delta(I-1)\big)\right),
$$
where $\Psi_i(l) := \frac{u_i'(l)}{u_i(l)}$ for $l\in(0,1)$.  We denote $\lambda:=\left(\lambda_{keep}+\frac{\lambda_{ubi}}{I}\big(1+\delta(I-1)\big)\right)$.  By It\^o's Lemma, we determine the dynamics
\begin{align*}
  dL_i &= \left(\frac{\alpha_i\lambda \mu_{w_i}}{\Psi_i'(L_i)} - \frac{\alpha_i^2\lambda^2 \sigma_{w_i}^2}{2}\cdot \frac{\Psi_i''(L_i)}{\left(\Psi_i'(L_i)\right)^3}\right)dt + \frac{\alpha_i\lambda \sigma_{w_i}}{\Psi_i'(L_i)} dB,\\
  d\sL_i &= \tfrac{1}{\alpha_i}d \left(\log u_i(L_i)\right)\\
  &=\left(\lambda \mu_{w_i} \cdot\frac{\Psi_i(L_i)}{\Psi_i'(L_i)} - \frac{\alpha_i \lambda^2 \sigma_{w_i}^2}{2} \cdot \frac{\Psi_i''(L_i)\Psi_i(L_i)}{(\Psi_i'(L_i))^3}+\frac{\alpha_i\lambda^2\sigma_{w_i}^2}{2\Psi'_i(L_i)}\right) dt + \frac{\lambda\sigma_{w_i}\Psi_i(L_i)}{\Psi'_i(L_i)}dB,\\
  d\big(w_i L_i) &= \left(\mu_{w_i}L_i +\frac{\alpha_i\lambda\sigma_{w_i}^2+\mu_{w_i}\Psi_i(L_i)}{\Psi'_i(L_i)} - \frac{\alpha_i \lambda \sigma_{w_i}^2}{2}\cdot \frac{\Psi_i(L_i)\Psi''_i(L_i)}{(\Psi_i'(L_i))^3} \right)dt \\
  & \quad \quad + \left(\sigma_{w_i}L_i + \frac{\sigma_{w_i}\Psi_i(L_i)}{\Psi'_i(L_i)}\right) dB, 
\end{align*}
which by the consistency requirement of Definition \ref{def:eq}, leads to
\begin{align*}
  d\eps_\Sigma &= (\lambda_{keep}+\lambda_{ubi})\sum_{i=1}^I d\big(w_iL_i\big).\end{align*}
By the invertibility of $\Psi_i:(0,1)\rightarrow\R$, we can express $L_i = \Psi_i^{-1}(\alpha_i \lambda w_i)$, allowing us to express all of the terms and dynamics coefficients above as functions of $\lambda$ or wage rates.

For each $i=1,\ldots, I$, since $\Psi_i^{-1}$ is strictly increasing and $w_i$ is positive, $L_i$ is strictly increasing in the parameter $\lambda =\left(\lambda_{keep}+\frac{\lambda_{ubi}}{I}\big(1+\delta(I-1)\big)\right)$ for $\lambda>0$.  Thus, $L_i$, $\sL_i$, $w_iL_i$, $\eps_i(L_i)$, and $\eps_\Sigma$ are all strictly increasing in $\lambda$ for $\lambda>0$.  For $\lambda_{ubi}>0$, when $\delta>-\frac{1}{I-1}\left(1+ \tfrac{\lambda_{keep}\cdot I}{\lambda_{ubi}}\right)$, $\lambda$ is positive.

For these comparative statics, fix
$\lambda_{total}:=\lambda_{keep}+\lambda_{ubi}$ and vary $\lambda_{ubi}$, so that
$\lambda_{keep}=\lambda_{total}-\lambda_{ubi}$ and $\lambda = \lambda_{total}+\frac{I-1}{I}(\delta-1)\lambda_{ubi}$.  Since $L_i=\Psi_i^{-1}(\alpha_i\lambda w_i)$ is increasing in $\lambda$, increasing $\lambda_{ubi}$ increases labor participation and aggregate income when $\delta>1$, lowers them when $\delta<1$, and has no effect when $\delta=1$, provided that $\lambda>0$.



\subsection{Stock Market Effects}\label{section:stock-market-effects}
How does the equilibrium financial market respond to changes in universal basic income policy?  We investigate these changes through the market-price-of-risk and interest rate, which we define via a state price deflator.  A state price deflator is a strictly positive It\^o process $\xi$ with dynamics
$$
  d\xi = -\xi \big(\kappa \,dB + r\, dt\big), \quad \xi_0>0,
$$
such that $\left(\xi A + \int_0^\cdot \xi_s ds\right)$ and $\left(\xi S + \int_0^\cdot \xi_s D_s ds\right)$ are local martingales.  Given a state price deflator $\xi$ for the equilibrium financial market, $\xi$'s dynamics are constrained so that its market-price-of-risk $\kappa$ and interest rate $r$ satisfy
$$
  \mu_S = \kappa \sigma _S \quad \text{and} \quad r = \mu_A - \kappa\sigma_A.
$$
In our setting, we have
\begin{align*}
  \kappa &= \alpha_\Sigma( \sigma_D + \sigma_\eps - \sigma_{\sL}),\\
  r &= \rho_\Sigma + \alpha_\Sigma( \mu_D + \mu_\eps - \mu_{\sL}) - \frac12 \kappa^2.
\end{align*}

We call $r$ an interest rate, even though the financial market does not necessarily have an interest rate since it is unclear whether $\sigma_S$ is nonzero for all time with probability one.  However, should there be an interest rate in the usual sense, it would be our $r$.

As seen above in Section \ref{section:labor-mkt}, for each $i=1,\ldots, I$, the terms $L_i$, $\sL_i$, and $w_iL_i$ are strictly increasing in $\lambda$ and $w_i$ since $\Psi_i^{-1}$ is strictly increasing and $w_i$ is strictly positive, but their dynamics coefficients are not.  Even though we have closed-form expressions for all dynamics coefficients in terms of input parameters, these terms are nonmonotone in the input parameters.  The market price of risk and interest rate depend on those dynamics coefficients, and thus, they depend in a complex, nonmonotone manner on the input parameters:  $\lambda_{keep}$, $\lambda_{ubi}$, and $\delta$.

\begin{figure}[t]
	\begin{center}
	$
			\begin{array}{ccc}
			\includegraphics[width=0.40\textwidth]{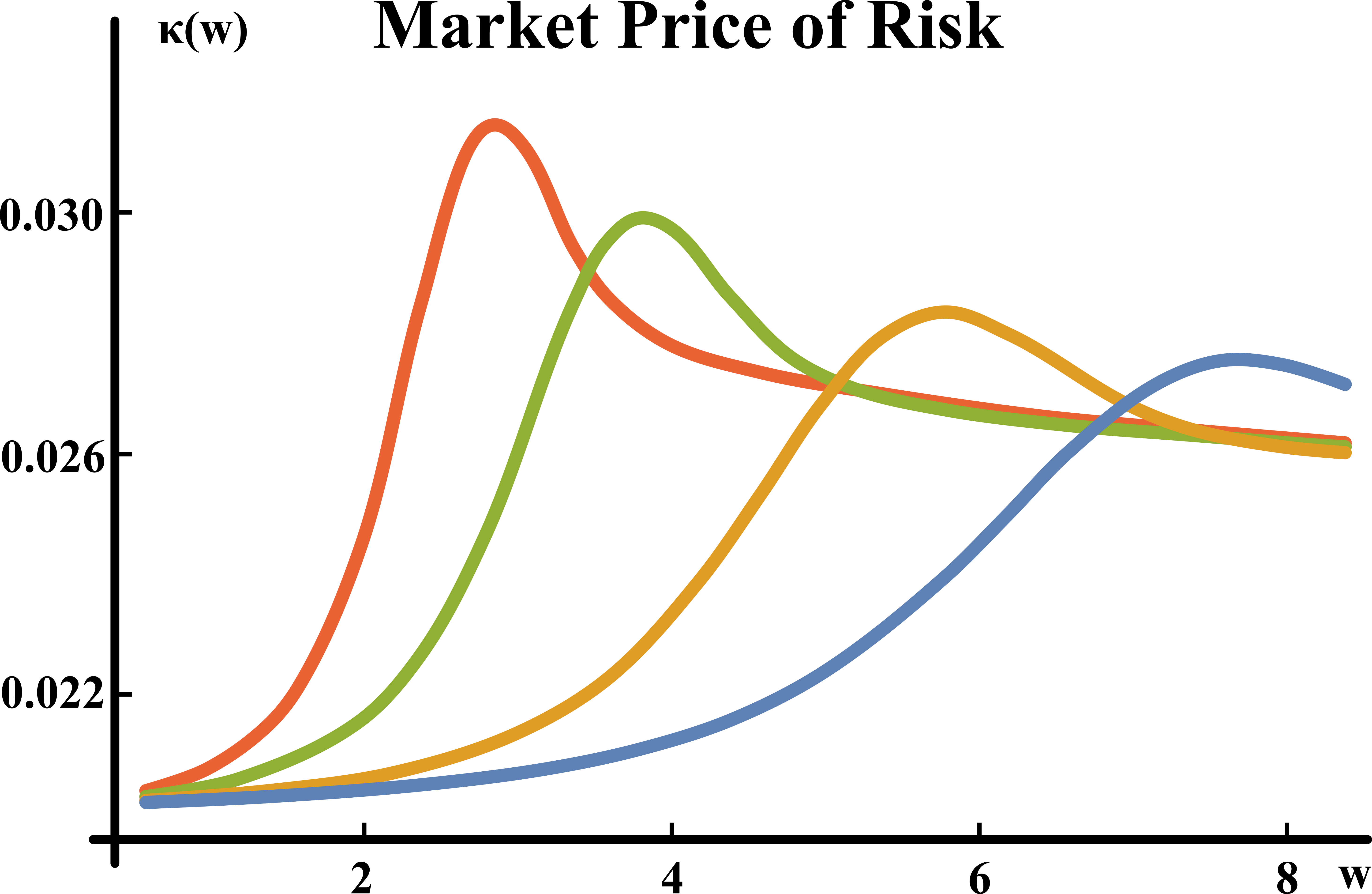} &
			\includegraphics[width=0.40\textwidth]{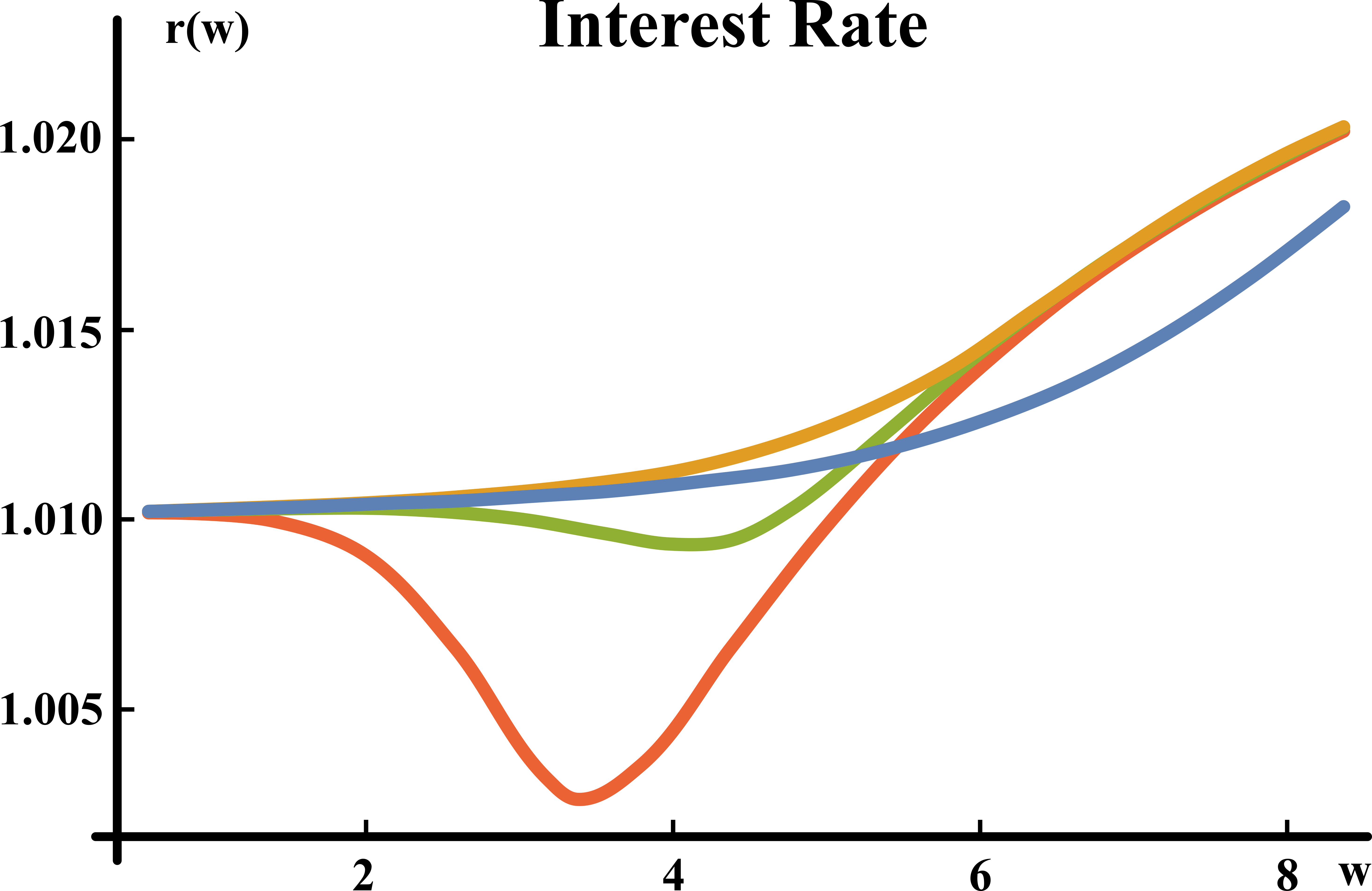}  &
			\includegraphics[width=0.15\textwidth]{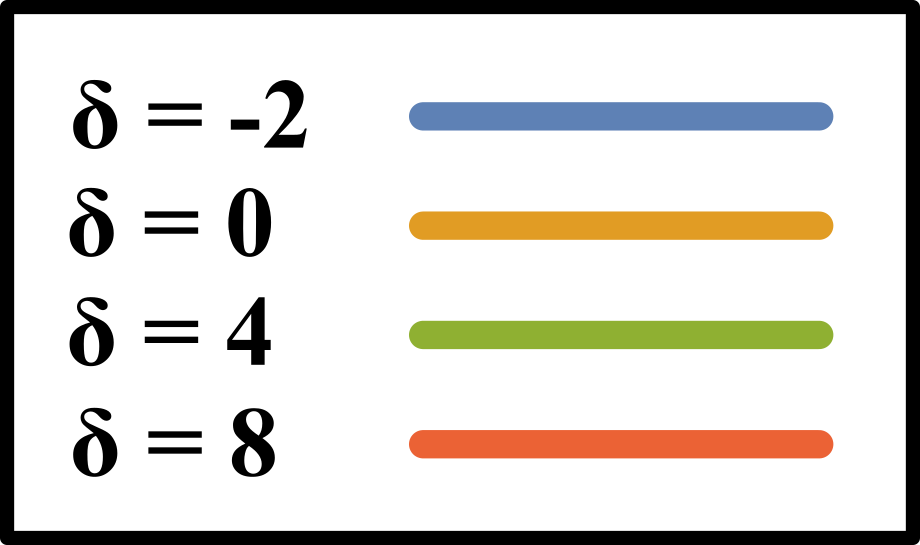} \\
			\end{array}$
	\end{center}
	\caption{Graphs of the market price of risk and interest rate as functions of the wage rate $w$ for varying $\delta$.  The common set of parameters is $\lambda_{keep}=0.7$, $\lambda_{ubi}=0.2$, $I=2$, $\alpha_1=\alpha_2 = 0.2$, $\rho_1=\rho_2=1$, $\beta_1=-0.02$, $\beta_2=-0.01$, $\gamma_1 = -0.8$, $\gamma_2 = -1$, $\mu_{w_1} = 0.1$, $\mu_{w_2}=0.2$, $\sigma_{w_1} = 0.1$, $\sigma_{w_2}=-0.05$, $\mu_D = 0.1$, and $\sigma_D = 0.2$.  }
	\label{fig:kappa-r}
\end{figure}

Figure~\ref{fig:kappa-r} plots the market price of risk and interest rate as functions of the wage rate across a common set of parameters $\lambda_{keep}=0.7$, $\lambda_{ubi}=0.2$, $I=2$, $\alpha_1=\alpha_2 = 0.2$, $\rho_1=\rho_2=1$, $\beta_1=-0.02$, $\beta_2=-0.01$, $\gamma_1 = -0.8$, $\gamma_2 = -1$, $\mu_{w_1} = 0.1$, $\mu_{w_2}=0.2$, $\sigma_{w_1} = 0.1$, $\sigma_{w_2}=-0.05$, $\mu_D = 0.1$, and $\sigma_D = 0.2$.  The initial wage rate is taken to be the same for both agents in the plots, and the graphs plot the market price of risk and interest rate as functions of the shared initial wage rate.  The $\beta_1$, $\beta_2$, $\gamma_1$, and $\gamma_2$ values are chosen so that $|\beta_1|\ll|\gamma_1|$ and $|\beta_2|\ll |\gamma_2|$ in order to approximate a labor utility function that is strictly decreasing in the labor proportion (and hence strictly increasing in leisure), as in Remark~\ref{rmk:ui}.  The other parameter values are chosen to illustrate the nontrivial behavior of the interest rate $r$ and market price of risk $\kappa$, where $r$ and $\kappa$ are observed to be nonmonotone in the initial wage rate for some values of $\delta$. Nontrivial behavior is also observed with respect to the influence parameter, $\delta$.  For example, as $\delta$ increases in Figure~\ref{fig:kappa-r}, both the market price of risk and interest rate exhibit nonmonotone behavior with respect to $\delta$.

%
%

\subsection{Welfare}
We study the welfare of the economy using aggregate certainty equivalents in the simple case of constant wage rates and a dividend stream with a drift and volatility that are deterministic.  Assume in this subsection that $w_1, \ldots, w_I$ are positive constants.  Since wage rates are constant, the optimal labor proportions $L_1, \ldots, L_I$ and the corresponding income streams are also constant.  We also assume in this section that $\mu_D = \mu_D(t)$ and $\sigma_D=\sigma_D(t)$ are measurable functions of time with $\sigma_D(t)\neq 0$ for all $t\in[0,T]$.

For $i=1,\ldots, I$, the certainty equivalent for agent $i$ is defined as the constant value $CE_i$ such that
$$
  \int_0^T U_i(t,CE_i,L_i) dt + U_i\big(T,CE_i, L_i\big) = \E\left[\int_0^T U_i(t,c_{i,t},L_i) dt + U_i\big(T,X_{i,T}+\eps_{i}(L_i), L_i\big)\right],
$$
where $c_i$, $X_i$, and $L_i$ are optimal for agent $i$ in \eqref{def:optimization}.  The time dependence on $\eps_i$ and $L_i$ is dropped due to constant wage rates.

The certainty equivalent represents the constant consumption stream level that an agent is willing to exchange for her optimal (stochastic) consumption stream.

Proposition \ref{prop:welfare} provides a representation for the aggregate certainty equivalent in the constant wage rate case.  The aggregate certainty equivalent measures the welfare in the economy.  The proof of Proposition~\ref{prop:welfare} is in Section~\ref{section:proofs}.

\begin{proposition}\label{prop:welfare}
Assume that the wage rates are positive constants and that
$\mu_D,\sigma_D:[0,T]\to\mathbb R$ are deterministic bounded
measurable functions, with $\sigma_D(t)\neq0$ for every
$t\in[0,T]$. Define
\[
    K_i(T):=\int_0^T e^{-\rho_i t}\,dt+e^{-\rho_iT}.
\]
Then agent $i$'s certainty equivalent is
\[
    CE_i
    =
    \frac{1}{\alpha_i}\log u_i(L_i)
    +
    \frac{1}{\alpha_i}\log K_i(T)
    +
    \frac{X_{i,0}}{A_0}
    +
    Y_{i,0}.
\]
Consequently,
\[
    \sum_{i=1}^I CE_i
    =
    1+D_0+\eps_{\Sigma,0}
    -\frac{a_0}{\alpha_\Sigma}
    +
    \sum_{i=1}^I
    \frac{1}{\alpha_i}\log K_i(T).
\]
\end{proposition}

Within the welfare representation, due to the constant wage rates, the universal basic income parameters ($\lambda_{keep}$, $\lambda_{ubi}$, and $\delta$) only appear in $\eps_{\Sigma,0}$, which as in Section~\ref{section:labor-mkt}, is increasing in $\lambda = \left(\lambda_{keep} + \frac{\lambda_{ubi}}{I}\left(1+\delta(I-1)\right)\right)$, so long as $\lambda>0$.

Holding $\lambda_{total}=\lambda_{keep}+\lambda_{ubi}$ fixed as above, we have
  $$
    \eps_{\Sigma,0} = \lambda_{total}\sum_{i=1}^I w_iL_i,
  $$
so the policy parameters affect aggregate welfare through the labor choices $L_i$. Consequently, if $\lambda>0$, then increasing $\lambda_{ubi}$ raises welfare when $\delta>1$, lowers welfare when $\delta<1$, and leaves welfare unchanged when $\delta=1$.

\subsection{Model Extremes}\label{section:extremes}
We compare two extreme scenarios in our model:  pure capitalism versus a social market economy.  Under pure capitalism, agents keep all of their individual earnings ($\lambda_{keep}=1$), no earnings are redistributed ($\lambda_{ubi} = 0$), and the influence parameter $\delta$ is irrelevant.    A social market economy assumes all individual earnings are redistributed ($\lambda_{keep}=0$, $\lambda_{ubi}=1$).  In a social market economy, even though agents give up all of their individual earnings to the economy at large, they perceive this usurping of their income differently in their labor choice problem, \eqref{epsi} and \eqref{def:optimization}, depending on the influence parameter $\delta$.  Recall that throughout Section~\ref{section:effects}, wage rates are assumed to be positive.

In pure capitalism, labor choices using Definition~\ref{def:Li} are determined by
$$
  \frac{u_i'(L_i)}{u_i(L_i)} = \alpha_i w_i,
$$
whereas in a social market economy, labor choices are determined by
$$
  \frac{u_i'(L_i)}{u_i(L_i)} = \alpha_i w_i \left(\frac{1}{I} + \delta\cdot\frac{I-1}{I}\right).
$$
In a social market economy, labor market participation depends on social sentiment, measured by the influence parameter $\delta$.  In social market economies, labor force participation is greater (less) than that of pure capitalism when $\delta>1$ ($\delta<1$).

In both economies, the total income of the economy is $\eps_\Sigma  = (\lambda_{keep}+\lambda_{ubi}) \sum_{i=1}^I w_i L_i = \sum_{i=1}^I w_i L_i$ by Definition \ref{def:eq} item \ref{eq:perceptions}.  Each $L_i$ fluctuates depending on $\delta$ (social market economy only), $\lambda_{keep}$, and $\lambda_{ubi}$.  Since wages are positive, social market economies produce higher (lower) aggregate income than pure capitalism when $\delta>1$ ($\delta<1$).  Therefore, the social market economy only dominates pure capitalism in terms of labor market participation and aggregate income when social sentiment is large with $\delta>1$.

\section{Proofs}\label{section:proofs}

\begin{proof}[Proof of Proposition \ref{prop:Li}]
Let $i\in \{1,\ldots, I\}$ be given.  For $l\in(0,1)$, we define $\Psi_i(l) := \frac{u'_i(l)}{u_i(l)}$ and denote
$$
  \lambda = \lambda_{keep} + \frac{\lambda_{ubi}}{I}(1+\delta(I-1)).
$$
Then, the process $L_i$ satisfies $\Psi_i(L_i) = \alpha_i\lambda w_i$ and has dynamics $dL_i = \mu_{L_i}dt + \sigma_{L_i}dB$ with
\begin{equation}\label{Li-dynamics}
  \mu_{L_i} = \frac{\alpha_i\lambda \mu_{w_i}}{\Psi_i'(L_i)} - \frac{\alpha_i^2\lambda^2 \sigma_{w_i}^2}{2}\cdot \frac{\Psi_i''(L_i)}{\left(\Psi_i'(L_i)\right)^3}
  \quad \text{ and } \quad 
  \sigma_{L_i} = \frac{\alpha_i\lambda \sigma_{w_i}}{\Psi_i'(L_i)}.
\end{equation}
The functions $l\mapsto \frac{1}{\Psi_i'(l)}$ and $l\mapsto \frac{\Psi_i''(l)}{(\Psi_i'(l))^3}$ are uniformly bounded for $l\in(0,1)$ since they are continuous on $(0,1)$ and
\begin{align*}
  \lim_{l\rightarrow 0+} \frac{1}{\Psi_i'(l)} &= \lim_{l\rightarrow 1-} \frac{1}{\Psi_i'(l)} = 0,\\
  \lim_{l\rightarrow 0+} \frac{\Psi''_i(l)}{(\Psi_i'(l))^3} &= \lim_{l\rightarrow 1-} \frac{\Psi''_i(l)}{(\Psi_i'(l))^3} = 0.
\end{align*}
Since $\mu_{w_i}$ and $\sigma_{w_i}$ are uniformly bounded, we conclude that $\mu_{L_i}$ and $\sigma_{L_i}$ are also uniformly bounded.

Next, we verify the uniform boundedness of $\mu_{\sL_i}$ and $\sigma_{\sL_i}$, where we recall that $\sL_i = \frac{1}{\alpha_i} \log u_i(L_i)$.  We have
\begin{align*}
  d \log u_i (L_i) &= \left(\alpha_i\lambda \mu_{w_i} \cdot\frac{\Psi_i(L_i)}{\Psi_i'(L_i)} - \frac{\alpha_i^2 \lambda^2 \sigma_{w_i}^2}{2} \cdot \frac{\Psi_i''(L_i)\Psi_i(L_i)}{(\Psi_i'(L_i))^3}+ \frac12\cdot\frac{\alpha_i^2\lambda^2\sigma_{w_i}^2}{\Psi'_i(L_i)}\right) dt\\
  & \quad\quad\quad\quad
  + \frac{\alpha_i\lambda\sigma_{w_i}\Psi_i(L_i)}{\Psi'_i(L_i)}dB.
\end{align*}
Continuity of $l\mapsto \frac{\Psi_i(l)}{\Psi'_i(l)}$ and $l\mapsto \frac{\Psi_i''(l)\Psi_i(l)}{(\Psi_i'(l))^3}$ for $l\in(0,1)$ and the limits
\begin{align*}
  \lim_{l\rightarrow 0+} \frac{\Psi_i(l)}{\Psi_i'(l)} &= \lim_{l\rightarrow 1-} \frac{\Psi_i(l)}{\Psi_i'(l)} = 0,\\
  \lim_{l\rightarrow 0+} \frac{\Psi''_i(l)\Psi_i(l)}{(\Psi_i'(l))^3} &= \lim_{l\rightarrow 1-} \frac{\Psi''_i(l)\Psi_i(l)}{(\Psi_i'(l))^3} = 0,
\end{align*}
imply that the functions $l\mapsto \frac{\Psi_i(l)}{\Psi'_i(l)}$ and $l\mapsto \frac{-\Psi_i''(l)\Psi_i(l)}{(\Psi_i'(l))^3}$ are uniformly bounded on $(0,1)$.  Together with the boundedness of $\mu_{w_i}$ and $\sigma_{w_i}$, we conclude that $\mu_{\sL_i}$ and $\sigma_{\sL_i}$ are uniformly bounded.


By \eqref{epsi}, the uniform boundedness of $\mu_{\eps_i}$ and $\sigma_{\eps_i}$ relies on the boundedness of the drift and volatility of $w_i L_i$.  The product $w_i L_i$ is an It\^o process with drift coefficient $\big(w_i\mu_{L_i} + \mu_{w_i}L_i + \sigma_{w_i}\sigma_{L_i}\big)$ and volatility coefficient $\big(w_i\sigma_{L_i} + \sigma_{w_i}L_i\big)$.  By the measurability and boundedness of $\mu_{w_i}$, $\mu_{L_i}$, $\sigma_{w_i}$, $\sigma_{L_i}$, and $L_i$, we only need to verify the boundedness of $w_i\sigma_{L_i}$ and $w_i\mu_{L_i}$.

For $w_i\sigma_{L_i}$, we use Definition~\ref{def:Li} for $L_i$ and the dynamics \eqref{Li-dynamics} to obtain
$$
  w_i\sigma_{L_i} = \frac{\alpha_i\lambda \sigma_{w_i} w_i}{\Psi'_i(L_i)} = \frac{\sigma_{w_i}\Psi_i(L_i)}{\Psi'_i(L_i)},
$$
where the above boundedness of $l\mapsto \frac{\Psi_i(l)}{\Psi'_i(l)}$ for $l\in(0,1)$
implies uniform boundedness of $w_i\sigma_{L_i}$.

Similarly for $w_i\mu_{L_i}$, we use Definition~\ref{def:Li} and dynamics \eqref{Li-dynamics} to obtain
$$
  w_i\mu_{L_i} = \frac{\mu_{w_i}\Psi_i(L_i)}{\Psi'_i(L_i)} - \frac{\alpha_i \lambda \sigma_{w_i}^2}{2}\cdot \frac{\Psi_i(L_i)\Psi''_i(L_i)}{(\Psi_i'(L_i))^3},
$$
where the above boundedness of $l\mapsto \frac{\Psi_i(l)}{\Psi'_i(l)}$ and $l\mapsto\frac{\Psi_i(l)\Psi''_i(l)}{(\Psi_i'(l))^3}$ for $l\in(0,1)$ imply that $w_i\mu_{L_i}$ is uniformly bounded.

\end{proof}

\begin{proof}[Proof of Proposition \ref{prop:bsde}]
We consider the secondary BSDE system
\begin{align}
  \begin{split}\label{def:bsde2}
    da &= Z_a\, dB + \left(\rho_\Sigma +\alpha_\Sigma(\mu_D+\mu_\eps-\mu_\sL) -\frac{1}{2}\big(\alpha_\Sigma(\sigma_D+\sigma_\eps-\sigma_\sL)-Z_a\big)^2- \exp(-a)\right)dt,\\
    d\overline{Y}_i &= \overline{Z}_i\,dB + \frac{1}{\alpha_i}\left(-\rho_i + \alpha_i(\mu_{\sL_i}-\mu_{\eps_i}) + \frac{1+a+\alpha_i\overline{Y_i}}{\exp(a)}-\frac{1}{2}\big(\alpha_\Sigma(\sigma_D+\sigma_\eps-\sigma_\sL)-Z_a\big)^2\right.\\
    &\quad\quad\quad\quad\quad\quad\left.\phantom{\frac{(Y_i)}{A}}+\alpha_i\left(\overline{Z_i}+\sigma_{\eps_i} - \sigma_{\sL_i}\right)\big(\alpha_\Sigma(\sigma_D+\sigma_\eps-\sigma_\sL)-Z_a\big)\right)dt,\\
  a_T &= 0, \quad \overline{Y}_{i,T}= 0, \quad 1\leq i\leq I.
  \end{split}
\end{align}
We will show that there exists an $\sS^\infty\times\bmo$ solution to \eqref{def:bsde2}.  Then for each $i=1,\ldots, I$, taking $Y_i := \overline{Y}_i +\eps_i(L_i) -\sL_i$ and $Z_i := \overline{Z}_i + \sigma_{\eps_i} - \sigma_{\sL_i}$ gives us the solution $\big((a, Z_a), (Y_i, Z_i)_{1\leq i\leq I}\big)$ to the BSDE system \eqref{def:bsde}.  The uniform boundedness of $\sigma_{\eps_i}$ and $\sigma_{\sL_i}$ established in Proposition~\ref{prop:Li} ensures that $Z_i\in\bmo$ if and only if $\overline{Z}_i\in\bmo$.

Since the $(a, Z_a)$ equation decouples, we prove the existence of its solution first. For each $N\geq 1$, we consider a truncated equation,
\begin{align}\label{bsde-N} \tag{BSDE${ }_N$}
\begin{split} 
  da^{(N)} &= Z_a^{(N)}\, dB + \left(\rho_\Sigma +\alpha_\Sigma(\mu_D+\mu_\eps-\mu_\sL) -\frac{1}{2}\big(\alpha_\Sigma(\sigma_D+\sigma_\eps-\sigma_\sL)-Z_a^{(N)}\big)^2
  \right.\\
  &\left.\quad\quad\quad\quad\quad\quad  \phantom{\frac12}-\exp(-\max(a^{(N)},-N))\right)dt,\\
  a^{(N)}_T &= 0.
\end{split}\end{align}
Proposition 2 of Tevzadze~\cite{T08SPA} implies that there exists a unique $\sS^\infty\times\bmo$ solution $\big(a^{(N)}, Z_a^{(N)}\big)$.  Since $\big(a^{(N)}-\int_0^\cdot \left(\rho_\Sigma + \alpha_\Sigma\left(\mu_D + \mu_\eps - \mu_\sL\right)\right)\big)$ is a supermartingale, for each $t\in[0,T]$, we have
$$
  a^{(N)}_t \geq \E\left[a^{(N)}_T - \int_t^T\big( \rho_\Sigma+ \alpha_\Sigma\left(\mu_{D,s} + \mu_{\eps,s} - \mu_{\sL,s}\right)\big)ds\,\Big|\,\sF_t \right] \geq -C,
$$
where $C$ is a constant that is independent of $N$.  We derive $C$ by using bounds on $\mu_\eps$ and $\mu_\sL$ guaranteed by Proposition~\ref{prop:Li}.  For $N\geq C$ and $x\geq -C$, we have $-\exp(-\max(x, -N)) = -e^{-x}$, so the lower bound on $a^{(N)}$ guarantees that $\big(a^{(N)}, Z^{(N)}_a\big)$ solves ($\text{BSDE}_C$).  
Moreover, uniqueness of $\big(a^{(N)}, Z^{(N)}_a\big)$ as a solution to ($\text{BSDE}_C$) implies that $\big(a^{(N)}, Z^{(N)}_a\big) = \big(a^{(C)}, Z^{(C)}_a\big)$ for every $N\geq C$ and $\big(a^{(C)}, Z^{(C)}_a\big)$ solves
\begin{equation}\label{bsde:a}
  da = Z_a\, dB + \left(\rho_\Sigma +\alpha_\Sigma(\mu_D+\mu_\eps-\mu_\sL) -\frac{1}{2}\big(\alpha_\Sigma(\sigma_D+\sigma_\eps-\sigma_\sL)-Z_a\big)^2- \exp(-a)\right)dt,
\end{equation}
with $a_T = 0$.  Thus, $\big(a^{(C)}, Z^{(C)}_a\big)$ is the unique $\sS^\infty\times\bmo$ solution to \eqref{bsde:a} with $a_T=0$.

Let $(a,Z_a)$ be the unique $\sS^\infty\times\bmo$ solution to \eqref{bsde:a} with $a_T=0$.  Most standard BSDE existence results require boundedness of the $(t,\omega)$-dependent driver terms.  For the driver of $(\overline{Y}_i, \overline{Z}_i)$ in \eqref{def:bsde2}, $Z_a\in\bmo$ cannot be guaranteed to be bounded.  We apply the linear BSDE results of Jackson and \v{Z}itkovi\'c~\cite{JZ22SICON} to obtain the existence and uniqueness of $(\overline Y_i, \overline Z_i)$ in
\begin{align*}
  d\overline{Y}_i &= \overline{Z}_i\,dB + \frac{1}{\alpha_i}\left(-\rho_i +\alpha_i(\mu_{\sL_i}-\mu_{\eps_i})+ \frac{1+a+\alpha_i\overline{Y_i}}{\exp(a)}-\frac{1}{2}\big(\alpha_\Sigma(\sigma_D+\sigma_\eps-\sigma_\sL)-Z_a\big)^2\right.\\
    &\quad\quad\quad\quad\quad\quad\left.\phantom{\frac{(Y_i)}{A}}+\alpha_i\left(\overline{Z_i}+\sigma_{\eps_i} - \sigma_{\sL_i}\right)\big(\alpha_\Sigma(\sigma_D+\sigma_\eps-\sigma_\sL)-Z_a\big)\right)dt\\
  \overline Y_{i,T} &= 0.
\end{align*}

Using the notation ``$\alpha$'' of Jackson and \v{Z}itkovi\'c~\cite{JZ22SICON}, we have that $\alpha = \exp(-a)$ is sliceable due to its boundedness.  Since our BSDE is of dimension one, Jackson and \v{Z}itkovi\'c~\cite{JZ22SICON} Proposition 2.4 and Corollary 2.10 (item 2) imply that there exists a unique solution $(\overline Y_i, \overline Z_i)\in\sS^\infty\times\bmo$.

Therefore, combining the existence and uniqueness for these uncoupled, one-dimensional BSDEs, we have that there exists a unique $\sS^\infty\times\bmo$ solution to \eqref{def:bsde2}, as desired.

\end{proof}

\begin{proof}[Proof of Theorem \ref{thm:main}]
  We fix $\lambda_{keep}\in[0,1]$, $\lambda_{ubi}\in[0,1-\lambda_{keep}]$, and $\delta\in\R$.  Let the labor proportions $L_i$ for $1\leq i\leq I$ be defined by Definition~\ref{def:Li} and labor response functions $\Lambda^i_j$ for $1\leq i, j\leq I$ with $j\neq i$ be defined by Definition~\ref{def:response-fn}.  By applying Proposition~\ref{prop:bsde}, we obtain the unique solution $\big((a,Z_a), (Y_i, Z_i)_{1\leq i\leq I}\big)$ to the BSDE system \eqref{def:bsde}.  Define the equilibrium price processes, strategies, and dynamics terms as in the statement of Theorem~\ref{thm:main}.
    
  
  We must check that all of the conditions from Definition \ref{def:eq} are satisfied.
  \ \\
  
  \noindent {\it Optimality.} Let $i\in\{1,\ldots, I\}$ be given.  We first verify that $(\pi_i,\theta_i,c_i,L_i)\in\sA_i$. Since
  $
    \sM^i(\pi_i,\theta_i,c_i,L_i)\equiv 0
  $,
  it remains to check items~\ref{adm:1} and~\ref{adm:dyn} of Definition~\ref{def:adm}.  From the definition of $\pi_i$,
\[
    \frac{\pi_i\sigma_S}{A}
    =
    \frac{1}{\alpha_i}
    (\kappa-\sigma_A-\alpha_iZ_i)
    \mathbf 1_{\{\sigma_S\neq0\}}
    \in\bmo.
\]
The dynamics of $X_i/A$ in \eqref{def:Xi} therefore have volatility in
$\bmo$ and drift in $\bmo^{1/2}$, since
$\kappa,\sigma_A,Z_i\in\bmo$ and
\[
    \frac{a+\alpha_i(Y_i-\eps_i(L_i)+\sL_i)}{\alpha_iA}
\]
is uniformly bounded. The energy inequalities and the
Burkholder--Davis--Gundy inequality imply that
$X_i/A\in\mathcal S^q$ for every $q>1$; see Kazamaki~\cite[(2.8) and the proof of
Theorem~2.2, p.~29]{Kaz94}.  Since $A$ and
$A^{-1}$ are uniformly bounded, it follows that $X_i\in\mathcal S^q$.
The regularity of $\eps_i(L_i)$ and
$Y_i-\eps_i(L_i)+\sL_i$, together with \eqref{def:ci}, similarly yields $\eps_i(L_i), Y_i, c_i\in\mathcal S^q$ for $q>1$.

Finally,
\[
\begin{aligned}
X_i\mu_A
&=
X_i\left(
\rho_\Sigma+\alpha_\Sigma(\mu_D+\mu_\eps-\mu_{\sL})
-\frac12\kappa^2+\kappa\sigma_A
\right),\\
\pi_i\mu_S
&=
\frac{A\kappa}{\alpha_i}
(\kappa-\sigma_A-\alpha_iZ_i)
\mathbf 1_{\{\sigma_S\neq0\}},\\
\pi_i\sigma_S
&=
\frac{A}{\alpha_i}
(\kappa-\sigma_A-\alpha_iZ_i)
\mathbf 1_{\{\sigma_S\neq0\}}.
\end{aligned}
\]
Putting this together with the boundedness of the remaining coefficients and H\"older's inequality gives
$$
  \E\int_0^T \Big(|X_{i,t}\mu_{A,t}| + |\pi_{i,t}\mu_{S,t}| + |\eps_{i,t}(L_{i,t})|
    + |c_{i,t}| + |X_{i,t}\sigma_{A,t}|^2 + |\pi_{i,t}\sigma_{S,t}|^2\Big) dt < \infty,
$$
and hence item~\ref{adm:1} holds.

Applying It\^o's product rule to $X_i=A(X_i/A)$, using the dynamics of $X_i/A$ in \eqref{def:Xi} and those of $A$, gives exactly the self-financing dynamics \eqref{eqn:self-financing}. The initial-wealth condition holds by the definition of $X_{i,0}$.  Thus, $(\pi_i, \theta_i, c_i, L_i)\in\sA_i$.

  Let $(\pi, \theta, c, L)\in\sA_i$ be given, and let $X=X^{c,L}= \pi S + \theta A$ denote its associated wealth process that satisfies the self-financing condition \eqref{eqn:self-financing} with the perceived income process $\eps_i(L)$.  We proceed with the proof of optimality in two steps.  First, we eliminate the labor control variable.  Then, we use a duality argument to conclude.
  
By self-financing, we have
  \begin{align}
  \begin{split}\label{eqn:X-dyn}
    dX^{c,L} 
    &= \pi (dS+D dt) + \theta(dA + dt) - \big((c-\eps_i(L)+\eps_i(L_i)) - \eps_i(L_i)\big) dt\\
    &= dX^{c-\eps_i(L)+\eps_i(L_i),L_i}.
  \end{split}
  \end{align}
Consequently, by keeping $(\pi, \theta)$ fixed, the transformed strategy $(\pi,\theta,c-\eps_i(L)+\eps_i(L_i),L_i)$ has the same initial wealth and the same wealth process:  $X = X^{c,L} = X^{c-\eps_i(L)+\eps_i(L_i),L_i}$.  Moreover, $\sM^i(\pi,\theta, c, L) = \sM^i(\pi, \theta, c-\eps_i(L)+\eps_i(L_i), L_i)$, so $(\pi, \theta, c-\eps_i(L)+\eps_i(L_i), L_i)\in\sA_i$.

At $t\in[0,T)$, $dt\otimes\bP$-a.e., we have
  \begin{align*}
    U_i&(t,c_t,L_t)
    = -\exp\big(-\rho_i t - \alpha_i c_t\big) u_i(L_t)\\
    &= -\exp\Big(-\rho_i t - \alpha_i \big(c_t-\eps_{i,t}(L_t) + \eps_{i,t}(L_{i,t})\big) + \alpha_i \big(\eps_{i,t}(L_{i,t})-\eps_{i,t}(L_t)\big)\Big) u_i(L_t)\\
    &\leq -\exp\Big(-\rho_i t - \alpha_i \big(c_t-\eps_{i,t}(L_t) + \eps_{i,t}(L_{i,t})\big)\Big) u_i(L_{i,t})\\
    &= U_i (t, c_t - \eps_{i,t}(L_t) + \eps_{i,t}(L_{i,t}), L_{i,t}).
  \end{align*}
The above inequality applies since $L_{i,t}$ is chosen to maximize $l\mapsto -e^{-\alpha_i \eps_{i,t}(l)}u_i(l)$ over $l\in(0,1)$.

At $T$, $\bP$-a.s., we have
  \begin{align*}
    U_i&\left(T, X^{c,L}_T + \eps_{i,T}(L_T), L_T\right)
    = -\exp\left(-\rho_i T - \alpha_i \left(X^{c,L}_T+\eps_{i,T}(L_T)\right)\right) u_i(L_T)\\
    &= -\exp\left(-\rho_i T - \alpha_i \left(X^{c,L}_T +\eps_{i,T}(L_T) - \eps_{i,T}(L_{i,T}) + \eps_{i,T}(L_{i,T})\right)\right) u_i(L_T)\\
    &\leq -\exp\Big(-\rho_i T - \alpha_i \left(X^{c,L}_T+\eps_{i,T}(L_{i,T})\right)\Big) u_i(L_{i,T})\\
    &= U_i \left(T, X^{c,L}_T + \eps_{i,T}(L_{i,T}), L_{i,T}\right)\\
    &= U_i \left(T, X^{c-\eps_i(L) + \eps_i(L_i),L_i}_T + \eps_{i,T}(L_{i,T}), L_{i,T}\right).
  \end{align*}
The above inequality applies since $L_{i,T}$ is chosen to maximize $l\mapsto -e^{-\alpha_i \eps_{i,T}(l)}u_i(l)$ over $l\in(0,1)$.  The last equality holds by \eqref{eqn:X-dyn}.

  Putting these calculations together yields
  \begin{align}
  \begin{split}\label{calc:L}
    \E&\left[\int_0^T U_i(t,c_t,L_t) dt + U_i(T, X_T + \eps_{i,T}(L_T), L_T)\right]\\
    &\leq \E\left[\int_0^T U_i(t,c_{t}-\eps_{i,t}(L_t)+\eps_{i,t}(L_{i,t}),L_{i,t}) dt + U_i(T, X_T + \eps_{i,T}(L_{i,T}), L_{i,T})\right].
  \end{split}
  \end{align}

Next, we use a duality argument to complete the proof of optimality.  Recall that $U_i(t,c,l) = -e^{-\rho_i t - \alpha_i c} u_i(l)$ for $t\in[0,T]$, $c\in\R$, and $l\in(0,1)$.  For \(y>0\), define the convex dual of the concave utility function
\(U_i(t,\cdot,l)\) by
\[
    \widetilde U_i(t,y,l)
    :=
    \sup_{c\in\mathbb R}
    \big(U_i(t,c,l)-yc\big).
\]
The functions $U_i$ and $\widetilde U_i$ are related by $U_i(t,c,l) = \widetilde U_i(t, \tfrac{\partial}{\partial c}U_i(t,c,l),l) + c \tfrac{\partial}{\partial c}U_i(t,c,l)$ for all $t\in[0,T]$, $c\in\R$, and $l\in(0,1)$.  Moreover, by the definition of $\xi_i$ in Definition~\ref{def:adm} item \ref{adm:3} and using BSDE \eqref{def:bsde}'s terminal conditions, we have
$$
  \xi_{i,T} = \frac{\partial}{\partial c} U_i(T, X_{i,T}+ \eps_{i,T}(L_{i,T}), L_{i,T}) 
  \quad \text{and} \quad 
  \xi_{i,t} = \frac{\partial}{\partial c} U_i(t, c_{i,t}, L_{i,t}), 
  \quad t\in[0,T]. 
$$
Putting these calculations together yields
\begin{equation}\label{eqn:duality-T}
  \widetilde U_i(T, \xi_{i,T}, L_{i,T}) + \xi_{i,T}(X_{i,T}+\eps_{i,T}(L_{i,T}))
  = U_i(T, X_{i,T} + \eps_{i,T}(L_{i,T}), L_{i,T})
\end{equation}
and
\begin{equation}\label{eqn:duality-t}
  \widetilde U_i(t, \xi_{i,t}, L_{i,t}) + \xi_{i,t}c_{i,t}
  = U_i(t, c_{i,t}, L_{i,t}), \quad t\in[0,T].
\end{equation}
Finally, we have
\begin{align*}
  \E&\left[\int_0^T U_i(t, c_t, L_t) dt + U_i(T, X_T + \eps_{i,T}(L_T), L_{T})\right]\\
  &\leq \E\left[U_i(T, X_T + \eps_{i,T}(L_{i,T}), L_{i,T})+\int_0^T U_i(t,c_{t}-\eps_{i,t}(L_t)+\eps_{i,t}(L_{i,t}),L_{i,t}) dt\right]\\
  &\leq \E\left[\widetilde U_i(T, \xi_{i,T}, L_{i,T}) + \xi_{i,T}\left(X_T + \eps_{i,T}(L_{i,T})\right) \phantom{\int_0^T \left(\widetilde U_i(\xi_{i})\right)}\right.\\
  &\left.\quad\quad\quad\quad+\int_0^T \left(\widetilde U_i(t, \xi_{i,t},L_{i,t}) + \xi_{i,t}\left(c_t-\eps_{i,t}(L_{t}) +\eps_{i,t}(L_{i,t})\right)\right) dt\right]\\
  &\leq \E\left[\widetilde U_i(T, \xi_{i,T}, L_{i,T}) + \xi_{i,T}\left(X_{i,T} + \eps_{i,T}(L_{i,T})\right) + \int_0^T \left(\widetilde U_i(t, \xi_{i,t},L_{i,t}) + \xi_{i,t}c_{i,t}\right) dt\right]\\
  &=\E\left[U_i(T, X_{i,T} + \eps_{i,T}(L_{i,T}), L_{i,T}) + \int_0^T U_i(t, c_{i,t}, L_{i,t}) dt\right].
\end{align*}

Here, the first inequality is due to \eqref{calc:L}.  The second inequality is due to the Fenchel-Young inequality, the definition of $\widetilde U_i$, and $\xi_i >0$.  The third inequality is due to the supermartingale property of
$$
  \sM^i(\pi, \theta, c-\eps_i(L)+\eps_i(L_i), L_i)
  = \xi_i(X-X_i) + \int_0^\cdot \xi_{i,s}\big((c_s-\eps_{i,s}(L_s)) - (c_{i,s} - \eps_{i,s}(L_{i,s}))\big)ds.
$$
The equality is due to the duality relationships \eqref{eqn:duality-T} and \eqref{eqn:duality-t}.  Therefore, the strategy $(\pi_i, \theta_i, c_i, L_i)$ with wealth $X_i$ is optimal for agent $i$'s problem~\eqref{def:optimization}.
  \ \\
  
    \noindent {\it Consistency.}  By Definition \ref{def:response-fn}, for $1\leq i,j\leq I$ with $i\neq j$, we have
  \begin{align}\label{calc:delta}
    \Lambda^i_j(L_i) = \delta\frac{w_i}{w_j}(L_i-L_i)+L_j = L_j,
  \end{align}
  as desired.  On $\{w_j\neq0\}$, the affine formula in Definition~3.2 gives $\Lambda_j^i(L_i)=L_j$. On $\{w_j=0\}$, the same identity holds by the separate definition of $\Lambda_j^i$.  Thus, $\Lambda_j^i(L_i)=L_j$ in either case, and the consistency condition holds.
  \ \\
  
  \noindent {\it On-equilibrium perceptions align with reality.} We have
  \begin{align*}
    \sum_{i=1}^I \eps_i(L_i)
    &= \sum_{i=1}^I \left(\lambda_{keep}w_i L_i +\frac{\lambda_{ubi}}{I}\left(w_i L_i + \sum_{j\neq i} w_j \Lambda^i_j(L_i)\right)\right) \quad \text{by \eqref{def:perceived-income}}\\
    &= \sum_{i=1}^I \left(\lambda_{keep}w_i L_i +\frac{\lambda_{ubi}}{I}\left(w_i L_i + \sum_{j\neq i} w_j L_j\right)\right) \quad \text{by \eqref{calc:delta}}\\
    &= \left(\sum_{i=1}^I w_iL_i\right)\left(\lambda_{keep} + \frac{\lambda_{ubi}}{I} +(I-1)\frac{\lambda_{ubi}}{I}\right)\\
    &= \left(\lambda_{keep}+\lambda_{ubi}\right)\sum_{i=1}^I w_i L_i.
  \end{align*}
\ \\
  
  \noindent {\it Market clearing.} For the stock market, on $\{\sigma_S = 0\}$, we have $\sum_{i=1}^I \pi_i = \sum_{i=1}^I \pi_{i,0-} = 1$, while when $\sigma_S\neq 0$,
  $$
    \sum_{i=1}^I \pi_i\sigma_S = \sum_{i=1}^I \frac{A}{\alpha_i}(\kappa - \sigma_A - \alpha_i Z_i) = \frac{A}{\alpha_\Sigma}(\kappa - \sigma_A - \alpha_\Sigma Z_\Sigma) = \sigma_S,
  $$
  by \eqref{def:pii} and \eqref{eqn:terms}.
  
  To see that we have clearing in the annuity and real goods market, we will prove that $\sum_{i=1}^I X_i = S+A$.  If $\sum_{i=1}^I X_i = S+A$ holds, then the annuity market clears using that $X_i = \theta_i A + \pi_i S$ for $i=1,\ldots,I$, and the real goods market clears by
  \begin{align*}
    \sum_{i=1}^I c_i &= \frac{1}{A}\sum_{i=1}^I X_i + \frac{a}{\alpha_\Sigma} + Y_\Sigma + \sL_\Sigma \\
    &= \frac{S}{A}+1 +\frac{a}{\alpha_\Sigma} + Y_\Sigma + \sL_\Sigma \quad \text{using that $\sum_{i=1}^I X_i = S+A$} \\
    &= 1+D + \sum_{i=1}^I \eps_i(L_i)\quad \text{by definition of $S$ in \eqref{def:prices}}.
  \end{align*}
  
  We proceed to prove that $\sum_{i=1}^I X_i = S+A$.  One can verify by using the dynamics of $\frac{X_i}{A}$ from \eqref{def:Xi} that
  \begin{align*}
    d\left(\sum_{i=1}^I \frac{X_i}{A}\right)
    &= \frac{1}{A}\left(\left(\eps_\Sigma-\frac{a}{\alpha_\Sigma} - Y_\Sigma - \sL_\Sigma\right) + \mu_S - \sigma_A\sigma_S\right)dt + \frac{\sigma_S}{A} dB \\
    &= \frac{1}{A}\left(\frac{S}{A} - D + \mu_S - \sigma_A\sigma_S\right)dt + \frac{\sigma_S}{A} dB\\
    &= d\left(\frac{S}{A}\right) = d\left(\frac{S+A}{A}\right).
  \end{align*}
  Since the processes $X_i$ are continuous and have $\sum_{i=1}^I X_{i,0} = \sum_{i=1}^I \left(\theta_{i,0-}A_0 + \pi_{i,0-}S_0\right) = A_0 + S_0$, we have that
  $$
    \frac{1}{A} \sum_{i=1}^I X_i = \frac{S+A}{A}.
  $$
   Thus, $\sum_{i=1}^I X_i = S+A$.
\end{proof}

Finally, we prove Proposition~\ref{prop:welfare}, which computes a representation for welfare in terms of the aggregate certainty equivalent in a market with constant, positive wage rates and deterministic dividend stream dynamics coefficients.

\begin{proof}[Proof of Proposition~\ref{prop:welfare}]
Since the wage rates are constant, the processes $L_i$,
$\eps_i(L_i)$, and $\sL_i$ are deterministic constants. Moreover,
because $\mu_D$ and $\sigma_D$ are deterministic, all coefficients
and terminal conditions in the BSDE system \eqref{def:bsde} are
deterministic. The system therefore has a deterministic solution with
zero martingale integrands. By uniqueness,
\[
    Z_a\equiv0
    \qquad\text{and}\qquad
    Z_i\equiv0,\quad i=1,\ldots,I.
\]
It follows that $\kappa=\alpha_\Sigma\sigma_D$, $\sigma_A = 0$, and $\sigma_S = A\sigma_D$.  Since $\sigma_D$ is nonzero, the definition of $\pi_i$ in \eqref{def:pii} gives
\[
    \frac{\pi_i\sigma_S}{A} =
    \frac{\alpha_\Sigma}{\alpha_i}\sigma_D.
\]

Fix $i\in\{1,\ldots,I\}$, and define
\[
    W_{i,t}
    :=
    -\exp\left(
        -\rho_i t
        -\alpha_i\left(
            \frac{X_{i,t}}{A_t}+Y_{i,t}
        \right)
    \right)
  \quad \text{and} \quad
    V_{i,t}
    :=
    \int_0^t U_i(s,c_{i,s},L_i)\,ds+W_{i,t}.
\]
Using the definition of $c_i$ in \eqref{def:ci}, together with
$A=e^a$ and $\alpha_i\sL_i=\log u_i(L_i)$, gives
\[
    U_i(t,c_{i,t},L_i)=\frac{W_{i,t}}{A_t}.
\]
Using It\^o's Lemma and the dynamics of $X_i/A$ and
$Y_i$ yields
\[
    dW_{i,t}
    =
    -\frac{W_{i,t}}{A_t}\,dt
    -\alpha_\Sigma\sigma_D(t)W_{i,t}\,dB_t
    \quad \text{and} \quad
    dV_{i,t} =
    -\alpha_\Sigma\sigma_D(t)W_{i,t}\,dB_t.
\]
Consequently, we have 
$$
  d\left(e^{\int_0^t \tfrac{1}{A_s}ds}W_{i,t}\right)
  = -\alpha_\Sigma \sigma_D(t) e^{\int_0^t \tfrac{1}{A_s}ds} W_{i,t} dB_t,
$$
and so $e^{\int_0^t \tfrac{1}{A_s}ds}W_{i,t} = W_{i,0} \sE\left(-\alpha_\Sigma \int_0^\cdot \sigma_D(s) dB_s\right)_t$.  Because $\sigma_D$ is bounded and deterministic, this stochastic
exponential is square-integrable. Since $A$ and $A^{-1}$ are
uniformly bounded,
\[
    V_{i,t} = V_{i,0} + 
    \int_0^t e^{-\int_0^s \tfrac{1}{A_{s'}}ds'}\,d\left(e^{\int_0^\cdot \tfrac{1}{A_{s'}}ds'}W_{i,\cdot}\right)_s
\]
is a martingale.

At the terminal time, $A_T=1$ and
$Y_{i,T}=\eps_i(L_i)-\sL_i$. Therefore,
$$
  W_{i,T}
  = -\exp\left(-\rho_iT -\alpha_i\bigl(X_{i,T}+\eps_i(L_i)-\sL_i\bigr)\right)
  = U_i\bigl(T,X_{i,T}+\eps_i(L_i),L_i\bigr).
$$
Thus,
\[
\begin{aligned}
&\mathbb E\left[
    \int_0^T U_i(t,c_{i,t},L_i)\,dt
    +
    U_i\bigl(T,X_{i,T}+\eps_i(L_i),L_i\bigr)
\right]\\
&\qquad
=\mathbb E[V_{i,T}]
= V_{i,0}
= -\exp\left(-\alpha_i\left(\frac{X_{i,0}}{A_0}+Y_{i,0}\right)\right).
\end{aligned}
\]

On the other hand, the definition of the certainty equivalent gives
$$
\int_0^T U_i(t,CE_i,L_i)\,dt +U_i(T,CE_i,L_i)
= -u_i(L_i)e^{-\alpha_iCE_i}K_i(T).
$$
Equating these two expressions gives
$$
  \log u_i(L_i) -\alpha_iCE_i +\log K_i(T)
    = -\alpha_i\left(\frac{X_{i,0}}{A_0}+Y_{i,0}\right),
$$
which yields the individual formula, $CE_i = \tfrac{1}{\alpha_i}\log u_i(L_i) + \tfrac{1}{\alpha_i}\log K_i(T) + \tfrac{X_{i,0}}{A_0} + Y_{i,0}$.

Summing over $i$ gives us that
\[
\begin{aligned}
\sum_{i=1}^I CE_i
&=
1+\frac{S_0}{A_0}
+Y_{\Sigma,0}
+\sL_\Sigma
+\sum_{i=1}^I
  \frac{1}{\alpha_i}\log K_i(T).
\end{aligned}
\]
By \eqref{def:prices}, $\frac{S_0}{A_0}
    +Y_{\Sigma,0}
    +\sL_\Sigma
    =
    D_0+\eps_{\Sigma,0}
    -\frac{a_0}{\alpha_\Sigma}$, which proves the desired result.
\end{proof}

\ \\
\noindent{\bf Funding:} No funding, grants, or other support was received to assist with the preparation of this manuscript.
\vskip 0.1in
\noindent{\bf Competing Interests:} The author has no relevant financial or non-financial interests to disclose.
\vskip 0.1in
\noindent{\bf Data availability statement:} No new data were created or analyzed in this study. Data sharing is not applicable to this article.


\bibliographystyle{plain}
\bibliography{finance_bib}

\end{document}